\documentclass[a4paper,reqno,11pt]{amsart}

\usepackage{amsfonts}
\usepackage{amssymb}
\usepackage{amsmath}
\usepackage[margin=1in]{geometry}
\usepackage{ esint }
\usepackage{enumitem}
\usepackage{physics}
\usepackage{tikz}
\usepackage{pgfplots}
\usepackage{dsfont}
\usepackage{dutchcal}
\usepackage{babel}
\usepackage{marginnote}
\usepackage{relsize}
\usepackage{bm}
\pgfplotsset{compat=newest}
\numberwithin{equation}{section}

\newcommand{\bbZp}{{\mathbb{Z}_{\geq0}}}

\newtheorem{thm}{Theorem}
\newtheorem{prop}[thm]{Proposition}
\newtheorem{lem}[thm]{Lemma}
\newtheorem{taggedlem}{Lemma}
\newenvironment{taglem}[1]
  {\renewcommand\thetaggedlem{#1}\taggedlem}
  {\endtaggedlem}
\newtheorem{cor}[thm]{Corollary}

\newtheorem{Prob}{Problem}

\theoremstyle{remark}
\newtheorem*{Rem}{Remark}
\newtheorem*{Rems}{Remarks}

\theoremstyle{definition}
\newtheorem{Ass}{Assumption}

\begin{document}
\title {Spectral and Dynamical contrast on highly correlated Anderson-type models.}
\author{Rodrigo Matos}
\email{matosrod@tamu.edu}
\address{Department of Mathematics, Texas A\&M University, College Station TX, 77843, USA.}
\author{Rajinder Mavi}
\author{Jeffrey Schenker}
\email{jeffrey@math.msu.edu}
\address{Department of Mathematics, Michigan State University, East Lansing MI 48823, USA.}

\maketitle
\begin{abstract}
{   We study spectral and dynamical properties of random Schr\"odinger operators  $H_{\mathrm{Vert}}=-A_{\mathbb{G}_{\mathrm{Vert}}}+V_{\omega}$ and $H_{\mathrm{Diag}}=-A_{\mathbb{G}_{\mathrm{Diag}}}+V_{\omega}$ on certain two dimensional graphs ${\mathbb{G}_{\mathrm{Vert}}}$ and ${\mathbb{G}_{\mathrm{Diag}}}$. Differently from the standard Anderson model, the random potentials are not independent but, instead, are constant along any vertical line, i.e $V_{\omega}(\bm{n})=\omega(n_1)$, for $\bm{n}=(n_1,n_2)$. In particular, the potentials studied here exhibit long range correlations. We present examples where geometric changes to the underlying graph, combined with high disorder, have a significant impact on the spectral and dynamical properties of the operators, leading to contrasting behaviors for the ``diagonal" and ``vertical" models.
 Moreover, the ``vertical" model exhibits a sharp phase transition within its (purely) absolutely continuous spectrum. This is captured by the notions of transient and recurrent components of the absolutely continuous spectrum, introduced by Avron and Simon in \cite{A-S}.}

\end{abstract}

\section{Introduction and Main results}
In this paper, we present and analyze examples of random Schr\"odinger operators for which contrasting dynamical and spectral behaviors can be observed. In comparison to the well established theory of Anderson localization, discussed below in detail, the systems studied here exhibit some form of long range correlations. Depending on the geometry of the underlying graph, the dynamical and spectral properties of the models can change significantly. Indeed, the first of the models described below, which we call the vertical model, exhibits purely absolutely continuous spectrum and a ballistic lower bound for the time averaged second moments of the position operator. Furthermore, its absolutely continuous spectrum splits into a transient and a recurrent component, in the sense of Avron and Simon \cite{A-S}. The transient spectrum for the vertical model is shown to appear only at the spectral edges and, for small values of a vertical hopping parameter, is much smaller (in the sense of Lebesgue measure) than the recurrent component. The notions of transient and recurrent absolutely continuous spectrum will be reviewed below in section \ref{transrecdef}.
On the other hand, the second model presented here, referred to as the diagonal model, exhibits dynamical localization and has pure point spectrum.

The nature of transport can be markedly different for strongly correlated potentials from what is familiar from the weakly correlated context. For instance, in \cite{M-S-J} two of us considered a system consisting of a particle in a random potential and a spin-$1/2$ which can flip only when the particle visits the origin.  This model can be viewed as an Anderson model on two lines connected at the origin, viewing the up and down spin states as distinct horizontal layers:
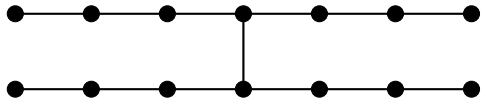
\vspace{0.5cm}
\begin{center}
\begin{tikzpicture}[scale=1.0]
\draw[fill=black] (-1,1) circle (3pt);
\draw[fill=black] (-1,2) circle (3pt);

\draw[fill=black] (0,1) circle (3pt);
\draw[fill=black] (0,2) circle (3pt);

\draw[fill=black] (1,1) circle (3pt);
\draw[fill=black] (1,2) circle (3pt);

\draw[fill=black] (2,1) circle (3pt);
\draw[fill=black] (2,2) circle (3pt);

\draw[fill=black] (3,1) circle (3pt);
\draw[fill=black] (3,2) circle (3pt);

\draw[fill=black] (4,1) circle (3pt);
\draw[fill=black] (4,2) circle (3pt);

\draw[fill=black] (5,1) circle (3pt);
\draw[fill=black] (5,2) circle (3pt);



(3,1)--(3,2);
\draw[thick]
(2,1)--(2,2);
 \draw[thick]
 (-1,2)--(0,2) --(1,2) -- (2,2) -- (3,2) --(4,2) -- (5,2);
 \draw[thick]
 (-1,1)--(0,1) --(1,1) -- (2,1) -- (3,1) --(4,1) -- (5,1);
\end{tikzpicture}
\end{center}
\vspace{0.5cm}
In this geometric picture, the potential is identical on the two layers, and thus has long range correlations in the graph metric.
In \cite{M-S-J} it was shown that resonant tunneling is compatible with correlated pure point spectrum since the model exhibits Green's function decay in the graph metric, has pure point spectrum, but its eigenfunctions are only localized in the particle position \cite[Theorems II.2
 and II.5]{M-S-J}. The present paper explores the consequences of correlations along the lines of those considered in \cite{M-S-J}, but of a longer range nature.  The key observation is that the \emph{geometry} of the hopping matters a great deal to dynamics in the presence of long range correlations.
 
\subsection{Overview of the models} \label{overview}
We now present a brief description of the graphs and random operators studied in this note. A detailed description is given below in section \ref{dfnsection}. 

\begin{figure}
\begin{center}
\begin{tikzpicture}[scale=0.5]
\def\width{12};
\def\height{4};
\foreach \m in {0,...,\width}
   \foreach \n in {0,...,\height}
        \draw[fill=black] (\m,\n) circle (3pt);
\foreach \n in {0,...,\height}
    \draw[fill=black] (\number\width +2,\n) circle (3pt);
\foreach \m in {0,...,\width}
    \draw[fill=black] (\m,\number\height +2) circle (3pt);
\draw[fill=black] (\number\width+2,\number\height +2) circle (3pt);
\foreach \n in {0,...,\numexpr\height -1}
     \draw[thick] (0,\n) -- (0,\n + 1);
\foreach \m in {0,...,\numexpr\width -1}
     \foreach \n in {0,...,\height}
         \draw[thick] (\m,\n) -- (\m+1,\n);
\foreach \m in {0,...,\numexpr\width-1}
    \draw[thick] (\m,\number\height +2) -- (\m+1,\number\height +2);
\foreach \n in {0,...,\height} 
     \draw[thick,dotted] (\width,\n) -- (\number\width+2,\n) -- (\number\width+4,\n);
\draw[thick,dotted] (\number\width,\number\height+2) -- (\number\width+2,\number\height + 2) -- (\number\width+4,\number\height + 2) ;
\draw[thick,dotted] (0,\height) -- (0,\number\height +2)-- (0,\number\height + 4);
\node at (-2,{0.5*(\number\height +4)}) {$\mathbb{G}_{\mathrm{Vert}}$\,\,\,\,\,\,};
\node at (0,-0.5){\fontsize{6}{9}\selectfont $(0,0)$};
\node at (-1, \number\height+2) {\fontsize{6}{9}\selectfont$(0,n)$};
\node at (\number\width+2,-0.5) {\fontsize{6}{9}\selectfont $(m,0)$};
\node at (\number\width+2,\number\height+2-0.5) {\fontsize{6}{9}\selectfont $(m,n)$};
\end{tikzpicture}
\end{center}
\caption{The graph $\mathbb{G}_{\mathrm{Vert}}$ }
\label{figure 1}
\end{figure}

Let $\mathbb{Z}_{\geq 0}=\{0,1,2,\ldots\}$. The graph $\mathbb{G}_{\mathrm{Vert}}$ has vertex set equal to $\mathbb{Z}_{\geq 0}\times \mathbb{Z}_{\geq 0}$ with nearest neighbor connections which are either horizontal or vertical (in the $y$-axis only); see figure \ref{figure 1}.  Contrasting to $\mathbb{G}_{\mathrm{Vert}}$ is the following family $\{\mathbb{G}_{\mathrm{Diag},\ell}\}^{\infty}_{\ell=0}$ of ``diagonal" graphs indexed by an integer $\ell\ge 0$. For $\ell=0$, let $\mathbb{G}_{\mathrm{Diag},0}$ denote the graph whose the vertex set lies on or below the diagonal of the first quadrant with nearest neighbors connected horizontally or through the diagonal $\{(n,n): n\in \mathbb{Z}_{\geq 0}\}$ (see Figure \ref{figure 2}).  For $\ell\ge 1$, the graph $\mathbb{G}_{\mathrm{Diag},\ell}$ is an ``interpolation" between $\mathbb{G}_{\mathrm{Vert}}$ and $\mathbb{G}_{\mathrm{Diag},\ell=0}$ obtained by alternating $\ell$ vertical connections among different layers with one ``diagonal" connection (see Figure \ref{figure 2} for the cases $\ell=1$ and $\ell=2$). Somewhat more precisely, the graph $\mathbb{G}_{\mathrm{Diag},\ell}$ has as its vertex set the portion of $\mathbb{Z}_{\geq 0}\times \mathbb{Z}_{\geq 0}$ on or to the right of the path $\mathcal{D}=\cup_{n\in \mathbb{Z}_{\geq 0}} \{\left(n,n(\ell+1)+r\right) : r=0,1,\ldots,\ell\}$, with vertices connected horizontally or along $\mathcal{D}$. To simplify notation, we often suppress the parameter $\ell$ in the discussion below, writing $\mathbb{G}_{\mathrm{Diag}}$ for $\mathbb{G}_{\mathrm{Diag},\ell}$ with the understanding that we are considering an arbitrary but fixed value of $\ell$.  A more detailed description of these graphs can be found in section \ref{dfnsection}.

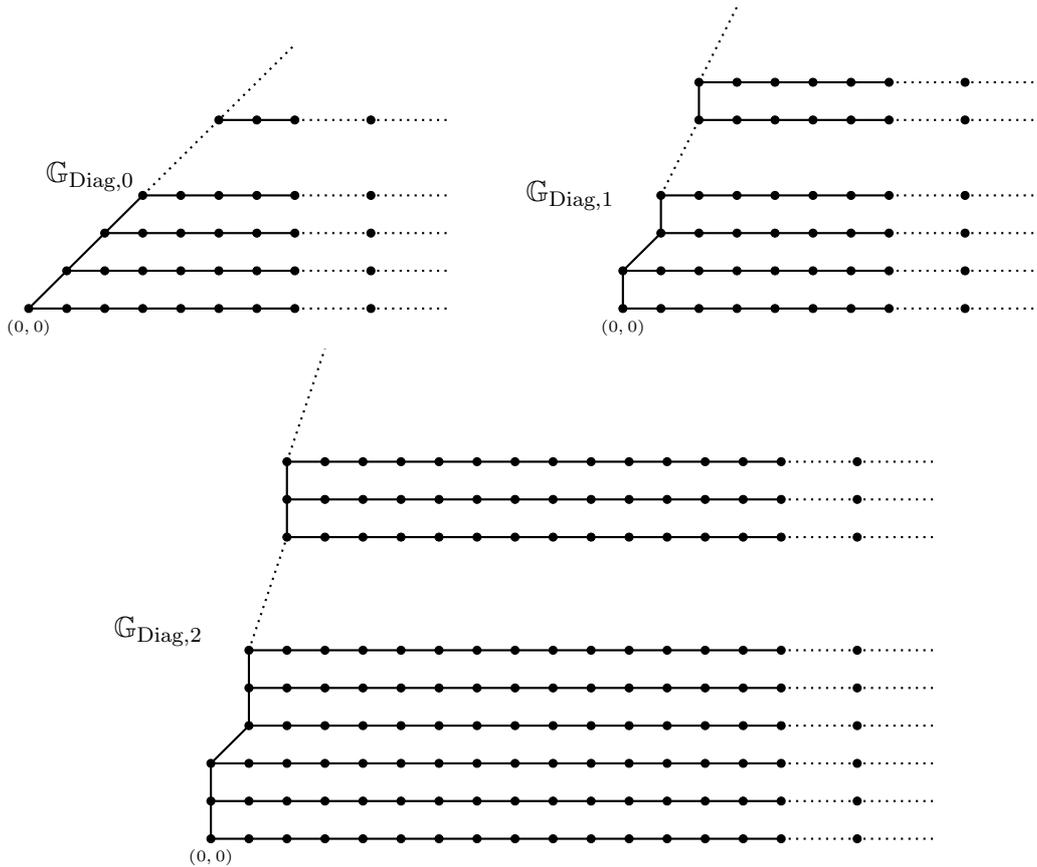
\begin{figure}
\begin{center}
\begin{tikzpicture}[scale=0.5]
    \def\width{7};
    \def\height{3};
    \foreach \n in {0,...,\height }
        \foreach \m in {\n,...,\width }
        \draw[fill=black] (\m,\n) circle (3pt);
    \foreach \m in {\numexpr\height+2,...,\width }
        \draw[fill=black] (\m,\numexpr\height+2) circle (3pt);
    \foreach \n in {0,...,\height }
        \draw[fill=black] (\numexpr\width+2,\n) circle (3pt);
    \draw[fill=black] (\numexpr\width +2, \numexpr\height +2) circle (3pt);
    \foreach \n in {0,...,\numexpr\height -1}
        \draw[thick] (\n,\n) -- (\n + 1 , \n + 1);
    \foreach \n in {0,...,\height}
        \foreach \m in {\n,...,\numexpr\width -1}
            \draw[thick](\m,\n)--(\m+1,\n);
    \foreach \m in {\numexpr\height+2,...,\numexpr\width-1}
        \draw[thick](\m,\numexpr\height+2) -- (\m+1,\numexpr\height+2);
    \foreach \n in {0,...,\numexpr\height}
        \draw[dotted,thick] (\numexpr\width,\n)--(\numexpr\width+2,\n) -- (\numexpr\width+4,\n);
    \draw[dotted,thick] (\width,\numexpr\height+2) -- (\width+2,\numexpr\height+2) -- (\width+4,\numexpr\height+2);
    \draw[dotted,thick] (\height,\height)--(\numexpr\height+2,\numexpr\height+2) -- (\numexpr\height+4,\numexpr\height+4);
    \node at ({0.5*(\number\height +1)},{0.5*(\number\height +4)}) {$\mathbb{G}_{\mathrm{Diag},0}$\,\,\,\,\,\,};
    \node at (0,-0.5) {\fontsize{6}{9}\selectfont $(0,0)$};
\end{tikzpicture}
\qquad
\begin{tikzpicture}[scale=0.5]
    \def\width{7};
    \def\height{1};
    \def\l{1};
    \def\fullheight{\the\numexpr((\height+1)*(\l+1)-1)\relax};
    \foreach \n in {0,...,\height }
        \foreach \k in {0,...,\l}
            \foreach \m in {\n,...,\width }
                \draw[fill=black] (\m,{((\l+1)*\n +\k)}) circle (3pt);
    \foreach \m in {\numexpr\height+1,...,\width }
        \foreach \k in {0,...,\l}
                \draw[fill=black] (\m,\numexpr\fullheight+\numexpr\l+1+\numexpr\k) circle (3pt);
    \foreach \n in {0,...,\numexpr\fullheight }
        \draw[fill=black] (\numexpr\width+2,\n) circle (3pt);
    \foreach \k in {0,...,\l}
        \draw[fill=black] (\numexpr\width +2, \numexpr\fullheight +\numexpr\l+1+\numexpr\k) circle (3pt);
    \foreach \n in {0,...,\numexpr\height }
        \foreach \k in {0,...,\numexpr\l-1}
            \draw[thick] (\n,{(\l+1)*\n+\k}) -- (\n ,  {(\l+1)*\n+\k+1});
    \foreach \n in {0,...,\numexpr\height -1 }
        \draw[thick](\n,{(\l+1)*(\n+1)-1)}) -- (\n +1 ,{(\l+1)*(\n+1)});
    \foreach \n in {0,...,\height }
        \foreach \k in {0,...,\l}
            \foreach \m in {\n,...,\numexpr\width -1}
            \draw[thick](\m,{((\l+1)*\n +\k)})--(\m +1,{((\l+1)*\n +\k)});
    \foreach \m in {\numexpr\height+1,...,\numexpr\width -1 }
        \foreach \k in {0,...,\l}
            \draw[thick] (\m,\numexpr\fullheight+\numexpr\l+1+\numexpr\k) -- (\m+1,\numexpr\fullheight+\numexpr\l+1+\numexpr\k);
    \foreach \n in {0,...,\height }
        \foreach \k in {0,...,\l}
            \draw[dotted,thick] (\numexpr\width,{(\l+1)*\n+\k})--(\numexpr\width+2,{(\l+1)*\n+\k}) -- (\numexpr\width+4,{(\l+1)*\n+\k});
     \foreach \k in {0,...,\numexpr\l-1}
            \draw[thick] (\height+1,\numexpr\fullheight+\numexpr\l+1+\numexpr\k) -- (\height+1 , \numexpr\fullheight+\numexpr\l+1+\numexpr\k +1);
    \draw[dotted,thick] (\height,\fullheight) -- (\height+1,\fullheight+\l+1);
    \draw[dotted,thick] (\height+1,\fullheight+\l+\l+1) -- (\height+2,\fullheight+\l+\l +\l+2);
     \foreach \k in {0,...,\l}
            \draw[dotted,thick] (\numexpr\width,\numexpr\fullheight+\numexpr\l+1+\numexpr\k)--(\numexpr\width+2,\numexpr\fullheight+\numexpr\l+1+\numexpr\k) -- (\numexpr\width+4,\numexpr\fullheight+\numexpr\l+1+\numexpr\k);
    \node at ({{\height-2}},{0.5*(\number\fullheight+3*\number\l)}) {$\mathbb{G}_{\mathrm{Diag},1}$\,\,\,\,\,\,};
    \node at (0,-0.5) {\fontsize{6}{9}\selectfont$(0,0)$};
\end{tikzpicture}
\qquad
\begin{tikzpicture}[scale=0.5] 
    \def\width{15};
    \def\height{1};
    \def\l{2};
    \def\fullheight{\the\numexpr((\height+1)*(\l+1)-1)\relax};
    \foreach \n in {0,...,\height }
        \foreach \k in {0,...,\l}
            \foreach \m in {\n,...,\width }
                \draw[fill=black] (\m,{((\l+1)*\n +\k)}) circle (3pt);
    \foreach \m in {\numexpr\height+1,...,\width }
        \foreach \k in {0,...,\l}
                \draw[fill=black] (\m,\numexpr\fullheight+\numexpr\l+1+\numexpr\k) circle (3pt);
    \foreach \n in {0,...,\numexpr\fullheight }
        \draw[fill=black] (\numexpr\width+2,\n) circle (3pt);
    \foreach \k in {0,...,\l}
        \draw[fill=black] (\numexpr\width +2, \numexpr\fullheight +\numexpr\l+1+\numexpr\k) circle (3pt);
    \foreach \n in {0,...,\numexpr\height }
        \foreach \k in {0,...,\numexpr\l-1}
            \draw[thick] (\n,{(\l+1)*\n+\k}) -- (\n ,  {(\l+1)*\n+\k+1});
    \foreach \n in {0,...,\numexpr\height -1 }
        \draw[thick](\n,{(\l+1)*(\n+1)-1)}) -- (\n +1 ,{(\l+1)*(\n+1)});
    \foreach \n in {0,...,\height }
        \foreach \k in {0,...,\l}
            \foreach \m in {\n,...,\numexpr\width -1}
            \draw[thick](\m,{((\l+1)*\n +\k)})--(\m +1,{((\l+1)*\n +\k)});
    \foreach \m in {\numexpr\height+1,...,\numexpr\width -1 }
        \foreach \k in {0,...,\l}
            \draw[thick] (\m,\numexpr\fullheight+\numexpr\l+1+\numexpr\k) -- (\m+1,\numexpr\fullheight+\numexpr\l+1+\numexpr\k);
    \foreach \n in {0,...,\height }
        \foreach \k in {0,...,\l}
            \draw[dotted,thick] (\numexpr\width,{(\l+1)*\n+\k})--(\numexpr\width+2,{(\l+1)*\n+\k}) -- (\numexpr\width+4,{(\l+1)*\n+\k});
     \foreach \k in {0,...,\numexpr\l-1}
            \draw[thick] (\height+1,\numexpr\fullheight+\numexpr\l+1+\numexpr\k) -- (\height+1 , \numexpr\fullheight+\numexpr\l+1+\numexpr\k +1);
    \draw[dotted,thick] (\height,\fullheight) -- (\height+1,\fullheight+\l+1);
    \draw[dotted,thick] (\height+1,\fullheight+\l+\l+1) -- (\height+2,\fullheight+\l+\l +\l+2);
     \foreach \k in {0,...,\l}
            \draw[dotted,thick] (\numexpr\width,\numexpr\fullheight+\numexpr\l+1+\numexpr\k)--(\numexpr\width+2,\numexpr\fullheight+\numexpr\l+1+\numexpr\k) -- (\numexpr\width+4,\numexpr\fullheight+\numexpr\l+1+\numexpr\k);
    \node at ({{\height-2}},{0.5*(\number\fullheight+3*\number\l)}) {$\mathbb{G}_{\mathrm{Diag},2}$\,\,\,\,\,\,};
    \node at (0,-0.5) {\fontsize{6}{9}\selectfont $(0,0)$};
\end{tikzpicture}
\caption{The graphs $\mathbb{G}_{\mathrm{Diag}}$   with $\ell=0$, $\ell=1$ and $\ell=2$.}
\label{figure 2}
\end{center}
\end{figure}

We now describe the operators of interest for this work, 
$$H_{{\mathrm{Vert}}}\ =\ -A_{\mathrm{Vert},\gamma} +V_{\omega} \quad \text{and} \quad 
H_{{\mathrm{Diag}}} \ =\ -A_{\mathrm{Diag},\gamma} +V_{\omega} \ . $$ Here
$A_{\sharp,\gamma}$, for $\sharp = \mathrm{Vert}$ or $\mathrm{Diag},\ell$,  denotes the weighted adjacency operator for the graph $\mathbb{G}_{\sharp}$, with hopping equal to $1$ along horizontal edges and to $\gamma>0$ along vertical edges ($\sharp=\mathrm{Vert}$) or diagonal edges ($\sharp=\mathrm{Diag}$).  We take the hopping equal to $1$ along vertical edges of $\mathbb{G}_{\mathrm{Diag}}$, although analogous results could be obtained for other values.  The operator $V_\omega$ is a  ``random potential'' of the form $\left(V_\omega\psi\right)(\bm{n})=V_\omega(\bm{n})\psi(\bm{n})$ with $V_\omega(\bm{n})$ random variables indexed by $\bm{n}\in \mathbb{G}_{\sharp}$.  The random potentials we consider depend only on the horizontal coordinate, and thus are perfectly correlated in the vertical direction. More precisely, we assume that
\begin{equation}\label{eq:Vomega} V_\omega(\bm{n}) \ = \ \omega(n_1) \quad \text{for} \quad \bm{n}=(n_1,n_2)\in \mathbb{G}_{\sharp} \ . \end{equation}

Our main assumption is the following
\begin{Ass}\label{ass:1} The random variables $\{\omega(n_1)\}_{n_1\in \mathbb{Z}_{\geq 0}}$ are non-negative, independent and identically distributed, with a common density $\rho$.  Furthermore, $\rho(v)>0$ for almost every $v\in [0,\omega_{\mathrm{max}}]$ and $\rho(v)$ vanishes for $v< 0$ and $v> \omega_{\mathrm{max}}$.  \end{Ass}
\begin{Rem} Note that $0\le \omega(n)\le \omega_{\mathrm{max}}$ almost surely and $\mathrm{Prob}[v-\epsilon<\omega(n)<v+\epsilon]>0$ for every $0\le v \le \omega_{\mathrm{max}}$ and $\epsilon >0$.
\end{Rem}

The restriction of  $H_{\sharp}$, $\sharp=\mathrm{Vert}$ or $\mathrm{Diag}$,  to a single layer $\mathbb{Z}_{\ge m_0} \times\{n_0\}$ of $\mathbb{G}_{\sharp}$ is a copy of the Anderson model on the half-line $\mathbb{Z}_{\ge m_0}=\{m\in \mathbb{Z} \ : \ m\ge m_0\}$:
\begin{equation}\label{eq:AM}
{\mathcal{h}}_{\mathrm{And}}^{(m_0)} \psi(m) \ = \ - \psi(m+1) - I[m\ge m_0 +1] \psi(m-1) + \omega(m) \psi(m) \ \quad m\ge m_0 \ .
\end{equation} 
The spectrum of this Anderson model is $\sigma(\mathcal{h}_{\mathrm{And}}^{(m_0)})=[-2,2+\omega_{\mathrm{max}}]$ almost surely (see, for example, \cite[Corollary 3.13]{A-W-B}). Furthermore, with probability one, $\mathcal{h}_{\mathrm{And}}^{(m_0)}$ exhibits Anderson localization and exponential dynamical localization (these concepts are reviewed below in \S\ref{sec:AL}). The operator $H_\sharp$ has \emph{identical} samples of the disorder on each layer and, in general, has some spectrum induced by the diagonal/vertical hopping that falls outside the interval $[-2,2+\omega_{\mathrm{max}}]$. As will become apparent from our main results, there are fundamental differences in the spectral and dynamical properties of the operators obtained by connecting these horizontal components in distinct ways.
\begin{Rem} The assumption that $\rho(v)>0$ for almost every $v\in [0,\omega_{\mathrm{max}}]$ is purely for convenience, as it allows us to identify the spectrum of the Anderson model with a single interval.  Most of what we do below would carry over to the more general case, even to unbounded potentials, with some modifications to Theorem \ref{transrecthm} in case there are additional boundaries in the spectrum. 
\end{Rem}

\subsection{Dynamical contrast between the vertical and diagonal models}
For $q>0$, the time-averaged $q$-moments
of a self-adjoint operator $H$ on $\ell^2(\mathbb{G}_\sharp)$ are defined by
\begin{equation}M^{q}_{T}(H,X_j)\ := \ \frac{2}{T}\int^{\infty}_{0}e^{\frac{-2t}{T}}\mathbb{E}\langle \delta_{\bm{0}},e^{itH}|X_j|^q e^{-itH}\delta_{\bm{0}}\rangle\,dt \ , \label{eq:Mqxj}
\end{equation}
 where $|X_j|$, $j=1,2$, acts as a multiplication operator on $\ell^{2}\left(\mathbb{G}_\sharp\right)$ via
 $\left(|X_j|^{q}\phi\right)(\bm{n}):=|n_j|^{q}\phi(\bm{n})$ for $\bm{n}=(n_1,n_2)\in \mathbb{G}_\sharp$. We also introduce 
 \begin{equation} M^{q}_{T}(H)\ := \ \frac{2}{T}\int^{\infty}_{0}e^{\frac{-2t}{T}}\mathbb{E}\langle \delta_{\bm{0}} ,e^{itH}|\bm{X}|^q e^{-itH}\delta_{\bm{0}}\rangle\,dt \label{eq:Mq}
\end{equation}
where $|\bm{X}|=|X_1|+|X_2|$ (note that, up to $q$-dependent constants, $M^{q}_T(H) \ \lesseqgtr  \ M^q_T(H,X_1)+M^q_T(X_2)$).

 Our first result concerns $H_{{\mathrm{Vert}}}$, where a combination of the symmetry and localization in the horizontal direction  induces ballistic transport in the vertical direction. 
\begin{thm}\label{thmsym} For all $\gamma>0$ there is $T_0>0$ such that the averaged moments satisfy
\begin{equation}\label{eqsym}
M^{q}_{T}\left(H_{\mathrm{Vert}},X_2\right)\geq C_{0}T^q
\end{equation}
for all times $T\geq T_0$ and some positive constant $C_{0}$ which depends on $\|\rho\|_{\infty}$ and $\gamma$.
\end{thm}
Perspectives and open problems related to Theorem \ref{thmsym} are discussed in  \S\ref{ballisticsubsec}. The proof of Theorem \ref{thmsym}, which appears in \S\ref{proofofthmsym} below, is based on the following outline:
\begin{enumerate}
    \item We show that $H_{\mathrm{Vert}}$ has purely absolutely continuous spectrum. This follows from the vertical structure of the graph $\mathbb{G}_{\mathrm{Vert}}$ along with the fact that the random potential depends only on the first coordinate.
    \item The Guarnieri bound \cite{Guarnieribound}, specialized to the two dimensional case, implies that $M^{q}_{T}\left(H_{\mathrm{Vert}} \right ) \geq C_{0}\left(T\right)^{\frac{q}{2}}.$ This is a general fact which does not rely on the randomness at all, only on the absolute continuity of the spectral measure $\mu_{\bm{0}}$ for $H_{\mathrm{Vert}}$ associated to $\delta_{\bm{0}}$.
    \item In Lemma \ref{horizontalprojbound} below, we show that $H_{\mathrm{Vert}}$ exhibits exponential dynamical localization in the horizontal direction, from which it follows that $M^{q}_{T}\left(H_{\mathrm{Vert}},\abs{X_1}\right) $ is bounded as $T\rightarrow \infty$. Therefore, it is possible to improve upon the Guarneri bound for $M^{q}_{T}\left(H_{\mathrm{Vert}}\right) $. 
    The intuitive idea is that transport may only occur in the vertical direction (since the horizontal direction essentially consists of an one-dimensional Anderson model).  Thus one should obtain a result consistent with the Guarnieri bound in one dimension, namely $M^{q}_{T}\left(H_{\mathrm{Vert}}\right) \approx M^{q}_{T}\left(H_{\mathrm{Vert}},\abs{X_2} \right)\geq C_{0}T ^{q}.$
\end{enumerate}
The above arguments are implemented through several technical steps in sections \ref{acsection}, \ref{floquetsection} and \ref{proofofthmsym}. We emphasize that theorem \ref{thmsym} is valid for all $\gamma>0$. 
We keep the vertical hopping parameter here for consistency, since it plays an important role in our results for $H_{\mathrm{Diag}}$.

Our second result concerns $H_{\mathrm{Diag}}=- A_{\mathrm{Diag},\gamma}+V_{\omega}$, which we show exhibits a strong form of dynamical localization in the horizontal direction:
\begin{thm}\label{thmdiag}
For each $\ell\in \mathbb{Z}_{\geq0}$ there exist $\gamma_0>0$ and $T_0>0$ such that whenever $\gamma<\gamma_0$, we have
\begin{equation}
    \mathbb{E} \left ( \sup_{t\in \mathbb{R}} \left | \langle \delta_{\bm{n}} \ , e^{-it H_{\mathrm{Diag}}} \delta_{\bm{m}} \rangle \right | \right ) \ \le \ 
    C_1 e^{-\nu |m_1-n_1|} \ ,
\end{equation}
with $C_1<\infty$ and $\nu >0$ depending on $\ell$ and $\gamma$.
\end{thm}
In this diagonal model, transport in the vertical direction is constrained by horizontal transport.  As a consequence we have a bound on all position moments:
\begin{cor}\label{Cor:bddmomentsdiag}
When $\gamma <\gamma_0$, with $\gamma_0$ as in Theorem \ref{thmdiag}, we have 
$
 \sup_T M^{q}_{T}\left(H_{\mathrm{Diag}} \right ) \ < \ \infty 
$
 for all $q$.
\end{cor}

Perspectives and open problems related to the above results are discussed in section \S\ref{ballisticsubsec}.
The main elements of the proof of Theorem \ref{thmdiag} are the following
\begin{enumerate}
    \item The model $H_{\mathrm{Diag}}$ exhibits, for small $\gamma$, exponential decay for the fractional moments of the  Green's function; see lemma \ref{Lemma:decay0n} below. This follows from an argument similar to the one in \cite[Theorem 6.3]{A-W-B} adapted to the present context. The main ingredient is Feenberg's loop-erased expansion for the Green's function (\cite[Theorem 6.2]{A-W-B}) combined with the geometry of the graph $\mathbb{G}_{\mathrm{Diag},\ell}$. In particular, it is crucial that for $n_2=k_2(\ell+1)+r_2$, $r_2\in\{0,\ldots,\ell\}$ and $n_2'>(k_2+1)(\ell+1)$, \emph{the restriction of $H_{\mathrm{Diag}}$ to $\ell^{2}\left(\mathbb{G}_{\mathrm{Diag},\ell}\cap(\mathbb{Z}_{\geq0}\times\{n_2'\})\right)$ is independent of $\omega(k_2)$, while the restriction of this operator to $\ell^{2}\left(\mathbb{G}_{\mathrm{Diag},\ell}\cap(\mathbb{Z}_{\geq0}\times\{n_2\})\right)$ depends on this variable.}
    \item Once decay of fractional moments of the Green's function is known, one expects to find upper bounds on the quantum dynamics. In the present context, we obtain exponential dynamical localization by a proof similar to that used in the context of continuum random Schr\"odinger operators  \cite[Theorem A1]{A-S-F-H}. 
    \end{enumerate}
The details of the above outline are completed in sections \ref{Greendecaysec} and \ref{proofthmdiagsec}.
\subsection{Spectral contrast between $H_{\mathrm{Vert}}$ and $H_{\mathrm{Diag}}$}
The following is a simple consequence of Theorem \ref{thmdiag} and the RAGE theorem. 
\begin{cor}\label{cormain} 
Whenever $\gamma<\gamma_0$, $H_{\mathrm{Diag}}$ has pure point spectrum with probability one.
\end{cor}

The spectral contrast between the two models is evident from the result below.

\begin{thm}\label{booleleb}
With probability one, $H_{\mathrm{Vert}}$ has simple, purely absolutely continuous spectrum.  Furthermore, $\delta_{\bm{0}}$ is a cyclic vector and the associated spectral measure $\mu_{\bm{0}}$ has a bounded density and is
supported on a set of Lebesgue measure  $4{\gamma}$.
\end{thm}
\begin{Rem}
That is, $d\mu_{\bm{0}}(E)=f(E)dE$ with $\sup_E f(E) < \infty$ and $|\{f>0\}|=4\gamma$.
\end{Rem}
A notable feature of Theorem \ref{booleleb} is that, for small values of  $\gamma$, the support of $\mu_{\bm{0}}$ has Lebesgue measure much smaller than the spectrum of $H_{\mathrm{Vert}}$, since the later contains the interval $[0,2+\omega_{\mathrm{max}}]$. While such behavior is necessary in systems exhibiting spectral localization (for which the support has measure zero), we are not aware of explicit examples of it in the context of random operators exhibiting transport and AC spectrum, as in the case of $H_{\mathrm{Vert}}$. As we shall see in the following section, this phenomenon is linked to the fact that $H_{\mathrm{Vert}}$ has \emph{recurrent AC} spectrum in $[0,2+\omega_{\mathrm{max}}]$.

The proof of Theorem \ref{booleleb} may be found in \S\ref{sec:booleleb} below. In calculating the Lebesgue measure of the support for $\mu_{\bm{0}}$, we make use of a generalization of \emph{Boole's identity} which is of independent interest. As we could not find a reference in the literature with the exact statement needed, we present the result here and give the details of the proof in the Appendix. Let $\mu$ be a finite Borel measure and let $F(z)=\int \frac{1}{u-z}\,d \mu(u) $ be its Borel transform, defined whenever $z\in \mathbb{C}^{+}$. Then the limit $$F(E+i0)=\lim_{\delta \to 0^{+}}F(E+i\delta)$$ exists and is finite for Lebesgue almost every $E$ and is furthermore real for almost every $E$ if $\mu$ is purely singular (see, e.g., \cite[Theorem 5.9.1]{Simon-Book}).
For the Borel transform of a singular measure there is a beautiful equality of Boole: 
\begin{prop}[Boole's identity \cite{Boole}]\label{boole}
Let $\mu$ be a finite, purely singular Borel measure on $\mathbb{R}$ and let $F(z)=\int \frac{1}{u-z} d\mu(u)$ be its Borel transform.  Then
 \begin{equation}\label{Booleeq}
 \big|\{E\in \mathbb{R}\,\,:F(E+i0)>t\}\big|=\frac{\mu\left(\mathbb{R}\right)}{t}.
 \end{equation}
\end{prop}
\begin{Rems}1) Here $\big|S\big|$ denotes the Lebesgue measure of $S$. 2) Boole's identity and its extensions have been
rediscovered or studied in various contexts by different authors (\cite{Loomis}, \cite{Stein}, \cite{Davis1}, \cite{Davis2}, \cite{Vinogradov}, \cite{D-J-L-S}, \cite{Poltora}). For further historical notes we refer to
\cite[Chapter 5]{Simon-Book} and \cite[Chapter 8]{A-W-B}.
\end{Rems}
\noindent To prove Theorem \ref{booleleb} we require the following generalization of Boole's identity:
\begin{prop}\label{booletype}
Let $\mu$ and $F$ be as in Proposition \ref{boole}. Then
 \begin{equation}\label{Boole's}
 \big|\{E\in \mathbb{R}\,\,:\alpha< E+F(E+i0)<\beta\}\big|=\beta-\alpha
 \end{equation}
 for every real $\alpha <\beta$.
\end{prop}

\subsection{Phase transition within $\sigma\left(H_{\mathrm{Vert}}\right)$}\label{transrecdef}
Our next result sheds light on Theorem \ref{booleleb}, providing further information on the dynamics $e^{-itH_{\mathrm{Vert}}}$ by describing the splitting of  the spectrum of $H_{\mathrm{Vert}}$ into transient and recurrent components, in the sense of Avron and Simon \cite{A-S}. This may be interpreted as a phase transition within the purely absolutely continuous spectrum of $H_{\mathrm{Vert}}$.

Before stating this result, it is useful to recall the notions of transient and recurrent AC spectrum. A key observation of \cite{A-S} was that the absolutely continuous subspace $\mathcal{H}_{\mathrm{ac}}$ of a self-adjoint operator $H$ can be further decomposed into its transient and recurrent subspaces.  The subspace ${\mathcal{H}}^{\mathrm{tac}}$ is defined to be the closure of the set of all $\psi\in {\mathcal{H}}^{\mathrm{ac}}$ such that, for all $N\in \mathbb{N}$
 \begin{equation}
 \left|\langle\psi,e^{-itH}\psi\rangle\right|=O\left(t^{-N}\right);
 \end{equation}
 such vectors are called ``transient vectors". By the Riemann-Lebesgue lemma, for any $\psi\in {\mathcal{H}}^{\mathrm{ac}}$ we have that $\lim_{t\to \infty}|\langle\psi , e^{-itH}\psi\rangle| = 0$.  For a transient vector $\psi\in {\mathcal{H}}^{\mathrm{tac}}$, we require the limit to converge faster than any inverse power of $t$. As a result, the Radon-Nikodym derivative $f_\psi(E)=\frac{d\mu_{\psi}}{dE}$ of the spectral measure associated to $\psi$ is a $C^\infty$ function (see \cite[Proposition 3.1]{A-S}).  The recurrent AC subspace $\mathcal{H}^{\mathrm{rac}}$ is defined to be the orthogonal complement of the transient space ${\mathcal{H}}^{\mathrm{tac}}$ within the AC subspace: $\mathcal{H}^{\mathrm{rac}} = \mathcal{H}^{\mathrm{ac}} \ominus \mathcal{H}^{\mathrm{tac}}$.  As explained in \cite{A-S}, one of their motivations is that in case $\mu_\psi=\chi_{C}\,dx$, where $C$ is a Cantor-like set of positive Lebesgue measure, the measure $\mu_\psi$ resembles a singular measure, despite its absolute continuity; and indeed this is a typical situation in which $\psi$ belongs to the recurrent subspace ${\mathcal{H}}^{\mathrm{rac}}$.

The transient and recurrent AC subspaces associated to a self-adjoint operator $H$ are seen to be invariant subspaces for $H$ \cite[Theorem 3.4]{A-S}. The transient and recurrent AC spectra of $H$, denoted $\sigma^{\mathrm{tac}}(H)$ and $\sigma^{\mathrm{rac}}(H)$ respectively, are the spectra of the restriction of $H$ to the corresponding subspaces, $\mathcal{H}^{\mathrm{tac}}$ and $\mathcal{H}^{\mathrm{rac}}$.

\begin{thm}\label{transrecthm}
For all $\gamma>0$ we have 
\begin{enumerate}[label=(\alph*)]
    \item $\sigma^{\mathrm{tac}}(H_{\mathrm{Vert}})$ is a non-deterministic closed subset of $[-2-2\gamma,-2] \cup [ 2+ \omega_{\mathrm{max}}, 2+\omega_{\mathrm{max}}+2\gamma].$
    \item $\sigma^{\mathrm{rac}}(H_{\mathrm{Vert}})=[-2,2+\omega_{\mathrm{max}}]$ \ .
\end{enumerate}
\end{thm}

The points $-2$ and $2+\omega_{\mathrm{max}}$ in $\sigma(H_{\mathrm{Vert}})$ are ``mobility edges'' separating two distinct types of spectra. The recurrent spectrum $\sigma^{\mathrm{rac}}(H_{\mathrm{Vert}})$ is equal to the bulk spectrum of the $1D$ Anderson model on the horizontal lines of $\mathbb{G}_{\mathrm{Vert}}$, whereas the transient spectrum $\sigma^{\mathrm{tac}}(H_{\mathrm{Vert}})$ falls outside the bulk spectrum. Further comments on the transient and recurrent subspaces of $H_{\mathrm{Vert}}$ are given in \S\ref{Sec:perspectives}. Theorem \ref{transrecthm} follows from Corollary \ref{ractacsplit} in \S\ref{Sec:transient} below.

\subsection{Organization of the paper} The remainder of this paper is organized as follows: \S\ref{sec:AL} consists of a brief review of Anderson localization, \S\ref{dfnsection} includes the precise definitions of the graphs $\mathbb{G}_{\mathrm{Vert}}$ and $\mathbb{G}_{\mathrm{Diag}}$, further perspectives and open problems are given in \S\ref{Sec:perspectives}. The proofs of results for $H_{\mathrm{Vert}}$ (Theorems \ref{thmsym}, \ref{booleleb} and \ref{transrecthm}) are given in \S\ref{Sec:proofofthmsym}.  The proof of Theorem \ref{thmdiag} (dynamical localization for $H_{\mathrm{Diag}}$) is given in \S\ref{sec:Hdiagproofs}. A proof of Proposition \ref{booletype}, a generalization of Boole's lemma, is given in  Appendix \ref{sec:Booleproof}.  In further appendices, we derive horizontal localization for $H_{\mathrm{Vert}}$ and review the harmonic analysis leading to boundedness of fractional moments of the Green's functions for $H_{\mathrm{Vert}}$ and $H_{\mathrm{Diag}}$.

\section{A short review of Anderson localization}\label{sec:AL}
We now discuss the relevant background on Anderson localization, as many of the specific results and different notions of localization will play a key role in the subsequent analysis. 

The effects of disorder on transport properties of quantum systems have drawn a significant amount of attention in the mathematics and physics communities since their introduction in 1958 in the celebrated paper \cite{A} by the physicist Anderson. The efforts to encode Anderson's claim that randomness localizes waves in disordered media into a rigorous mathematical statement and to obtain a
proof for it gave rise to a beautiful theory. 
For a more complete historical picture we refer to the survey
\cite{Stolz} and the book \cite{A-W-B}.

In the present paper we make extensive use of known bounds for the $1D$ Anderson model ${\mathcal{h}}_{\mathrm{And}}^{(m_0)}$ on the half-line $\ell^2(\mathbb{Z}_{\ge m_0})$, defined above  in eq.\ \eqref{eq:AM}.  More generally, the Anderson model may be defined on $\ell^2(\Omega)$, with $\Omega \subset \mathbb{Z}^d$, as the random operator $\mathcal{h}^{(\Omega)}=-A_\Omega + \lambda V_{\omega}$, where
\begin{enumerate}[label=(\roman*)]
\item{ $A_\Omega$ is the adjacency operator acting on $\varphi \in\ell^2\left(\Omega\right)$ through
    \begin{eqnarray*}\left(A_\Omega\varphi \right)(n)=\sum_{\substack{|m-n|_{1}=1 \\ m \in \Omega}}\varphi(m) \,,\quad \ n\in
    \Omega \, ,\quad |n|_{1}=|n_1|+\cdots+|n_d|.
\end{eqnarray*}}
\item{The random potential $V_{\omega}$ acts as a multiplication operator on $\ell^2\left(\Omega \right)$ via
$$\left(V_{\omega}\varphi\right)(n)=\omega(n)\varphi(n) \, , \quad n \in \Omega.
$$}
\item{$\omega=\{\omega(n)\}_{n\in \Omega}$ is a list of independent, identically distributed random variables.}
\item{$\lambda>0$ denotes the disorder strength.}
\end{enumerate}
Let $\{\delta_n\}_{n\in \Omega}$ be the canonical basis of
 $\ell^2\left(\Omega\right)$, with $\delta_n(m)=\delta_{mn}$, the Kronecker delta.
Dynamical localization is defined as averaged decay of the matrix elements
$|\langle\delta_n,e^{-it\mathcal{h}^{(\Omega)}}\delta_0 \rangle|
$, made explicit through a bound  such as
\begin{equation}\label{dynamicalloc}
\mathbb{E}\left(\sup_{t\in \mathbb{R}}|\langle\delta_n,e^{-it\mathcal{h}^{(\Omega)}}\delta_0 \rangle|\right)\leq Cr(n),
\end{equation}
where $C>0$ and
$\sum_{n\in \mathbb{Z}^d}r(n)<\infty.$
If the bound is obtained with $r(n)=e^{-\nu|n|}$, for some $\nu>0$, this is called \emph{exponential dynamical localization}. 
If $C_q := \sum_{n\in \mathbb{Z}^d}|n|^{q}r^2(n)<\infty,$  then
dynamical localization in the sense of (\ref{dynamicalloc}) implies the bound 
\begin{equation*}\mathbb{E}\left(\sup_{t\in \mathbb{R}}\,\langle \delta_0,e^{itH_{\omega}}|X|^q e^{-itH_{\omega}}\delta_0\rangle \right)\ \le\  C_q \ <\ \infty ,
    \end{equation*}
which in turn shows a bound on the disorder and time averaged moment (see eq. \eqref{eq:Mq})
\begin{equation}\label{finitemoments}
\mathbb{E}(M^{q}_{T}(H))\ \leq \ C_q \ < \ \infty \ .
\end{equation}
The inequality (\ref{finitemoments}) is a signature of localization whereas its counterpart, $M^{q}_{T}(H) \geq  CT^{\alpha} $
 for $\alpha>0$, indicates non-trivial transport which is called ballistic when $\alpha=q$ and diffusive in case $\alpha=q/2$.

There is a close relationship between dynamical localization, as in \eqref{dynamicalloc}, and decay of matrix elements of the Green's function
\begin{equation}\label{Greendfn} G^{(\Omega)}(n,m;z)=\langle \delta_n, (\mathcal{h}^{(\Omega)} -z)^{-1} \delta_m \rangle \end{equation} as $|n-m|\rightarrow \infty.$
For random potentials of the type considered here, with variables having an absolutely continuous distribution with a bounded density, a convenient signature of exponential localization is given by exponential decay of the \emph{fractional moments of the Green's function}, namely 
\begin{equation}\label{eq:fmdecay}
\sup_{E} \limsup_{\epsilon \rightarrow 0}\mathbb{E}\left ( \left |G^{(\Omega)}(n,m;E+i\epsilon)|^s \right | \right ) \ \le \ C_{\mathrm{And}} e^{-\mu_{\mathrm{And}}|n-m|} \ ,
\end{equation}
with $0<s <1$, $\mu_{\mathrm{And}} >0$ and $C_{\mathrm{And}}<\infty$. 
See \cite[Chapter 7]{A-W-B} for a full discussion of the relation between fractional moments and dynamical localization. The key fact for the purposes of the present paper is that:
\begin{quote}\emph{Eqs.\ \eqref{dynamicalloc} and \eqref{eq:fmdecay} hold for the one dimensional Anderson model $\mathcal{h}_{\mathrm{And}}^{(m_0)}$ on the half line $\mathbb{Z}_{\ge m_0}$ with non-constant random variables satisfying Assumption \ref{ass:1}.} \end{quote}
See, e.g., \cite[Chapter 12]{A-W-B} for further details.

More generally, in the one dimensional setting, exponential dynamical localization has been shown for any $\lambda>0$ whenever the support of the random variables $\{\omega(n)\}_{n\in \mathbb{Z}}$ contains at least two points. This is the result of many efforts, starting with \cite{K-S}; see also \cite{GMP} for the analysis of a related one-dimensional model. For singular distributions, complete spectral localization in one dimension was first showed in \cite{C-K-M}, and the recent works \cite{Dam} and \cite{J-Z} establish complete exponential dynamical localization. In dimension $d\geq2$, exponential dynamical localization has been proved at large disorder, meaning that $\lambda$ is taken sufficiently large, or at weak disorder at the edges of spectral bands, see \cite[Theorems 10.2 and 10.4 ]{A-W-B} for precise statements.

Finally, we recall the notion of \emph{spectral localization}. Associated to any self-adjoint operator $H$ on  a Hilbert space  $\mathcal{H}$, there is a
 decomposition $\mathcal{H}={\mathcal{H}}^{\mathrm{pp}}\oplus {\mathcal{H}}^{\mathrm{sc}}\oplus {\mathcal{H}}^{\mathrm{ac}}$ 
 into the pure point, singular continuous, and absolutely continuous sub-spaces, such that the spectral measure $\mu_\psi$ associated to a vector $\psi\in \mathcal{H}_\sharp$ is of the corresponding type (pure point for $\sharp=$pp, singular continuous for $\sharp=$sc, etc.).  The RAGE theorem (after Ruelle, Amrein, Georgescu and Enss; see \cite[Theorem 2.6]{A-W-B}) provides dynamical characterizations for these subspaces.  One of its consequences is that dynamical localization as in eq.\ (\ref{dynamicalloc}) implies that $\mathcal{h}^{(\Omega)}$ has pure point spectrum, meaning that $\mathcal{H}={\mathcal{H}}^{\mathrm{pp}}$ and the spectrum $\sigma(\mathcal{h}^{(\Omega)})$ is the closure of the set of eigenvalues for $\mathcal{h}^{(\Omega)}$.  When $\mathcal{h}^{(\Omega)}$ has pure point spectrum, we say that $\mathcal{h}^{(\Omega)}$ exhibits \emph{spectral localization.} If the associated eigenfunctions decay exponentially, the operator $\mathcal{h}^{(\Omega)}$ is said to exhibit \emph{exponential localization}. Neither spectral localization nor exponential localization implies dynamical localization in general; see, e.g., \cite{delrio1995,jitomirskaya2003}. We say that $\mathcal{h}^{(\Omega)}$ exhibits exponential decay of eigenfunction correlators when
 \begin{equation}
 \mathbb{E}\left (\sup_{|g|\leq 1} |\langle \delta_n,g(\mathcal{h}^{(\Omega)})\delta_0\rangle|\right)\ \leq \ Ce^{-\mu|n|} \
 \end{equation}
 holds for positive constants $C$ and $\mu$, where the above supremum is taken over all Borel measurable functions $g:\mathbb{R}\rightarrow \mathbb{C}$ bounded by one. Exponential decay of eigenfunction correlators follows from fractional moment localization \eqref{eq:fmdecay} and implies exponential dynamical localization and exponential localization. For a proof of these facts and more detailed statements we refer to \cite[Theorems 7.2 and 7.4]{A-W-B}.
 
\section{Definition of the Models}\label{dfnsection}
 We now proceed to define the graphs of interest for this work, starting with $\mathbb{G}_{\mathrm{Vert}}$ \textemdash \ see Figure \ref{figure 1} above. The vertex set of $\mathbb{G}_{\mathrm{Vert}}$ is given by 
\begin{equation}
{\mathcal{V}}_{\mathrm{Vert}}=\mathbb{Z}_{\geq 0}\times \mathbb{Z}_{\geq 0}
\end{equation}
where $\mathbb{Z}_{\geq 0}=\mathbb{N}\cup\{0\}.$
Given $\bm{m}=(m_1,m_2)$ and $\bm{n}=(n_1,n_2)$ in ${\mathcal{V}}_{\mathrm{Vert}}$, we write $\bm{m}\sim \bm{n}$ whenever $\bm{m}$ and $\bm{n}$ are connected by an edge. The edge set of $\mathbb{G}_{\mathrm{Vert}}$ is then given by $\bm{m}\sim \bm{n}$ such that either $\{m_2=n_2\,\, \mathrm{and}\,\, |m_1-n_1|=1\}$ or $\{m_1=n_1=0 \,\,\mathrm{and}\,\, |m_2-n_2|=1\},$ with $\bm{m},\bm{n}\in \mathbb{G}_{\mathrm{Vert}}$. 
Thus, the adjacency operator of $\mathbb{G}_{\mathrm{Vert}}$ is $X_{\mathrm{Vert}}+Y_{\mathrm{Vert}}$ with
\begin{equation}
    X_{\mathrm{Vert}}(\bm{m},\bm{n})=\left\{
\begin{array}{l}
 1\,\,\, \mathrm{if}\,\, m_2=n_2,\,\, |m_1-n_1|=1 \,\,\mathrm{and}\,\, \bm{m},\bm{n}\in {\mathcal{V}}_{\mathrm{Vert}}. \\ \\
0\,\,\, \mathrm{otherwise}.\\
\end{array}\right.
\end{equation}
and \begin{equation}
    Y_{\mathrm{Vert}}(\bm{m},\bm{n})=\left\{
\begin{array}{l}
 1\,\,\, \mathrm{if}\,\, m_1=n_1=0, \,\, |m_2-n_2|=1\,\,\mathrm{and}\,\, \bm{m},\bm{n}\in {\mathcal{V}}_{\mathrm{Vert}}.  \\ \\
0\,\,\, \mathrm{otherwise}.\\
\end{array}\right.
\end{equation}
We are interested in a weighted adjacency operator, namely $A_{\mathrm{Vert},\gamma}=X_{\mathrm{Vert}}+\gamma Y_{\mathrm{Vert}}.$
More explicitly,
\begin{equation}\label{AS}
A_{\mathrm{Vert},\gamma}(\bm{m},\bm{n})=\left\{
\begin{array}{l}
 \gamma\,\,\, \mathrm{if}\,\, \,\, m_1=n_1=0, \,\, |m_2-n_2|=1\,\,\mathrm{and}\,\, \bm{m},\bm{n}\in {\mathcal{V}}_{\mathrm{Vert}}.\\
1\,\,\, \mathrm{if}\,\, m_2=n_2,\,\, |m_1-n_1|=1 \,\,\mathrm{and}\,\, \bm{m},\bm{n}\in {\mathcal{V}}_{\mathrm{Vert}}.\\
0\,\,\, \mathrm{otherwise}.\\
\end{array}\right.
\end{equation}

We turn now to the graphs $\mathbb{G}_{\mathrm{Diag},\ell}$ for $\ell\in \mathbb{Z}_{\geq 0}$ \textemdash \ see Figure \ref{figure 2} above. For each $\ell\in \mathbb{Z}_{\geq 0}$, $\mathbb{G}_{\mathrm{Diag},\ell}$ is defined as follows. Its vertex set is 
\begin{equation}
    \mathcal{V}_{\mathrm{Diag},\ell} \ = \ \bigcup_{k\in \mathbb{Z}_{\ge 0}} \mathcal{V}_{\mathrm{Diag},\ell}^{(k)}
\end{equation}
with 
\begin{equation}
    \mathcal{V}_{\mathrm{Diag},\ell}^{(k)} \ = \ \left \{ (m_1,m_2)\in \mathbb{Z}_{\geq 0}\times\mathbb{Z}_{\geq 0}  \ \middle | \ m_1\ge k \ \text{ and } k(\ell+1) \le m_2 < (k+1)\ell+1 \right \}
\end{equation}
Two vertices $\bm{m}=(m_1,m_2)$ and $\bm{n}=(n_1,n_2)$ in $\mathbb{G}_{\mathrm{Diag}}$ are adjacent, $\bm{m}\sim \bm{n}$, if $(\bm{m},\bm{n})$ belongs any of the following three sets, which represent horizontal, vertical and ``diagonal" connections, respectively: 
\begin{equation*}
\mathcal{E}^{(1)}_{\mathrm{Diag}}=\left \{(\bm{m},\bm{n})  \ \middle |\ m_2=n_2\, , \,\,  |m_1-n_1|=1\,\, \mathrm{and}\,\, \bm{m},\bm{n}\in \mathcal{V}_{\mathrm{Diag},\ell} \right \} \ ,
\end{equation*}
\begin{equation*}
\mathcal{E}^{(2)}_{\mathrm{Diag}}=\left \{(\bm{m},\bm{n})  \ \middle | \ m_1=n_1=k\,\,\, |n_2-m_2|=1 \,\, \mathrm{and}\,\, \bm{m},\bm{n}\in \mathcal{V}_{\mathrm{Diag},\ell}^{(k)} \text{ for some } k\ge 0 \right \} \,\,
\end{equation*}
and
\begin{equation*}
\mathcal{E}^{(3)}_{\mathrm{Diag}}=\left \{ (\bm{m},\bm{n}) \, , \, (\bm{n},\bm{m}) \ \middle | \ \bm{n}=\bm{m}+(1,1)\ \text{with} \ \bm{m}=(k,k(\ell+1)) \text{ for some } k\ge 0 \right \} \ .
\end{equation*}
In Figure \ref{figure 3}, this decomposition is illustrated for $\mathbb{G}_{\mathrm{Diag,2}}$, with the connections of types $\mathcal{E}^{(1)}_{\mathrm{Diag}}$, $\mathcal{E}^{(2)}_{\mathrm{Diag}}$ and $\mathcal{E}^{(3)}_{\mathrm{Diag}}$ colored in black, red and blue, respectively.
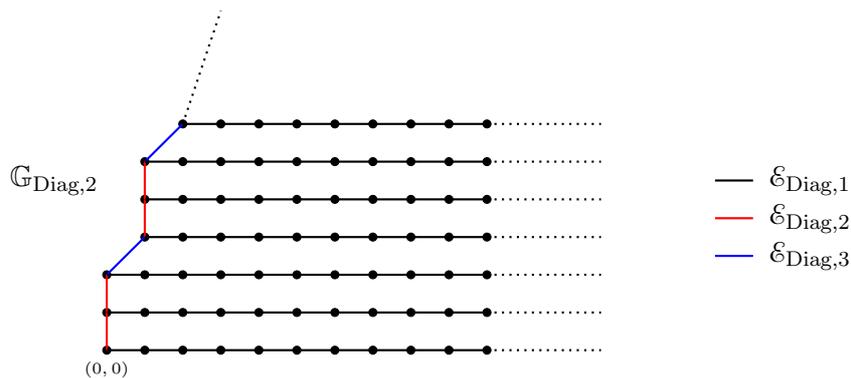
\begin{figure}
\begin{center}
\begin{tikzpicture}[scale=0.5]
 \def\width{10};
\def\height{1};
\def\l{2};
\def\fullheight{\the\numexpr((\height+1)*(\l+1)-1)\relax};
\foreach \n in {0,...,\height }
    \foreach \k in {0,...,\l}
        \foreach \m in {\n,...,\width }
            \draw[fill=black] (\m,{((\l+1)*\n +\k)}) circle (3pt);
\foreach \m in {\numexpr\height+1,...,\width }
    \draw[fill=black](\m,\fullheight+1) circle (3pt);
\foreach \n in {0,...,\numexpr\height }
    \foreach \k in {0,...,\numexpr\l-1}
        \draw[thick,red] (\n,{(\l+1)*\n+\k}) -- (\n ,  {(\l+1)*\n+\k+1});
\foreach \n in {0,...,\numexpr\height }
    \draw[thick,blue](\n,{(\l+1)*(\n+1)-1)}) -- (\n +1 ,{(\l+1)*(\n+1)});
\foreach \n in {0,...,\height }
    \foreach \k in {0,...,\l}
        \foreach \m in {\n,...,\numexpr\width -1}
        \draw[thick](\m,{((\l+1)*\n +\k)})--(\m +1,{((\l+1)*\n +\k)});
\foreach \m in {\numexpr\height+1,...,\numexpr\width - 1 }
    \draw[thick](\m,\fullheight+1) -- (\m +1,\fullheight+1);
\foreach \n in {0,...,\numexpr\fullheight +1 }
        \draw[dotted,thick](\numexpr\width,\n)--(\numexpr\width+3,\n);
\draw[dotted,thick] (\height+1,\fullheight+1) -- (\height+2,\fullheight+4);
\node at ({{\height-2}},{0.5*(\number\fullheight+4)}) {$\mathbb{G}_{\mathrm{Diag},2}$\,\,\,\,\,\,};
\node at (0,-0.5) {\fontsize{6}{9}\selectfont $(0,0)$};

\draw[thick] (\width+6,\number\fullheight -.5) -- (\width +7,\number\fullheight-.5);
\node at (\width + 7+1.5,\number\fullheight-.5)  {$\mathcal{E}_{\mathrm{Diag},1}$};
\draw[thick,red] (\width+6,\number\fullheight-1.5) -- (\width +7,\number\fullheight-1.5);
\node at (\width + 7+1.5,\number\fullheight-1.5)  {$\mathcal{E}_{\mathrm{Diag},2}$};
\draw[thick,blue] (\width+6,\number\fullheight-2.5) -- (\width +7,\number\fullheight-2.5);
\node at (\width + 7+1.5,\number\fullheight-2.5)  {$\mathcal{E}_{\mathrm{Diag},3}$};
\end{tikzpicture}
\end{center}
\caption{Types of edges in $\mathbb{G}_{\mathrm{Diag},\ell}$ with $\ell=2$.}
\label{figure 3}
\end{figure}
The adjacency operator of $\mathbb{G}_{\mathrm{Diag},\ell}$ is then $X_{\mathrm{Diag}}+Y_{\mathrm{Diag}}+D_{\mathrm{Diag}}$, with
\begin{equation}
    X_{\mathrm{Diag}}(\bm{m},\bm{n})=\left\{
\begin{array}{l}
 1\,\,\, \mathrm{if}\,\,  (\bm{m},\bm{n})\in \mathcal{E}_{\mathrm{Diag},1} \\ \\
0\,\,\, \mathrm{otherwise} ,\\
\end{array}\right.
\end{equation}

\begin{equation}
    Y_{\mathrm{Diag}}(\bm{m},\bm{n})=\left\{
\begin{array}{l}
 1\,\,\, \mathrm{if}\,\,(\bm{m},\bm{n})\in \mathcal{E}_{\mathrm{Diag},2}  \\ \\
0\,\,\, \mathrm{otherwise}, \\
\end{array}\right.
\end{equation}
and
\begin{equation}
    D_{\mathrm{Diag}}(\bm{m},\bm{n})=\left\{
\begin{array}{l}
 1\,\,\, \mathrm{if}\,\,(\bm{m},\bm{n})\in \mathcal{E}_{\mathrm{Diag},3}  \\ \\
0\,\,\, \mathrm{otherwise}.\\
\end{array}\right.
\end{equation}
We shall study a weighted version of this, namely $A_{\mathrm{Diag},\gamma}=X_{\mathrm{Diag}}+Y_{\mathrm{Diag}}+\gamma D_{\mathrm{Diag}}$.

\section{Perspectives and Open Problems}\label{Sec:perspectives}

 \subsection{On the ballistic bound of Theorem \ref{thmsym}}
 \label{ballisticsubsec}
The notion of ballistic transport employed here means that $M^q_{T}(H) \sim T^q$ with $M^q_{T}(H)$ as in \eqref{eq:Mq}. Note that this requires averaging over time and disorder. Such double averaging is important here, as our methods rely heavily on the Guarnieri bound \cite{Guarnieribound}, which requires time averaging, and on localization bounds in the horizontal direction which rely on disorder averaging. See sections \ref{floquetsection} and \ref{proofofthmsym} for further details. It is an interesting question whether ballistic bounds hold  without time averaging. 

We now mention a number of prior results on ballistic transport for various Schr\"odinger operators. A general ballistic upper bound, without time averaging, holds for discrete operators with finite range or exponentially bounded hopping terms, see, e.g., \cite[Appendix B]{A-W-Ballistic} for a proof. This bound corresponds to the single-particle version of the more general Lieb-Robinson bound \cite{L-R}.  In the context of random operators on a tree with independent single-site potentials, Aizenman and Warzel showed that absolutely continuous spectrum implies ballistic transport for time averaged moments, see \cite{A-W-Ballistic}. A ballistic upper bound for operators of the form $H=-\Delta+V$ on $L^2\left(\mathbb{R}^n\right)$, where $V$ is relatively bounded with respect to $\Delta$ with relative bound less than one, was obtained in \cite{R-S}.  Finally, the work \cite{K-L-S-S} establishes ballistic transport for certain limit periodic and quasi-periodic potentials in two dimensions.

 \subsection{ On the localization bound of Theorem \ref{thmdiag}}
As explained in the introduction, our proof of Theorem \ref{thmdiag} connects bounds on time-averaged moments to exponential decay of the Green's function fractional moments, see sections \ref{Greendecaysec} and \ref{proofthmdiagsec}. For independent potentials (more generally, potentials with a bounded conditional single-site distribution), decay of the Green's function fractional moments implies dynamical localization, see, for instance, \cite[Theorem A1]{A-S-F-H}. This perspective also allows to show dynamical localization in certain ``weakly" interacting systems as the ones considered by two of us in \cite{Mt-S}.
 
An alternative to fractional moments, the multiscale analysis technique, usually relies on the assumption of
\emph{independence at distance}, meaning that there exists a $R>0$ such that events based on boxes $\Lambda_{L_1}(m)$ are independent of
events based on boxes $\Lambda_{L_2}(n)$ if $\mathrm{dist}\left(\Lambda_{L_1}(m),\Lambda_{L_2}(n)\right)>R$. Here $\Lambda_{L}(m)=\{m'\in
\mathbb{Z}^d\,\,:|m-m'|_{\infty}<\frac{L}{2}\}.$
This assumption is not fulfilled in strongly correlated systems.

 \subsection{On the surface states of $H_{\mathrm{Vert}}$}
 The states $\psi \in \mathcal{H}^{\mathrm{tac}}$ are \emph{surface modes}, exponentially localized near the line $\{n_1=0\}$.
Such states are analogous to surface modes found in other disordered models \cite{JMdimtwo,J-M,JMMM,Liu} (though of a different dynamical character).  By way of contrast, the states in $\mathcal{H}^{\mathrm{rac}}$ are bulk states whose propagation, intuitively speaking, can be conceived of as resonant tunneling between states of the $1D$ Anderson model on the horizontal strips of the graph $\mathbb{G}_{\mathrm{Vert}}$, enabled by virtual transitions to the edge.

Both the surface and bulk modes can be formally described through separation of variables as generalized eigenfunctions of the form
$$\psi(n_1,n_2) \ = \ \sin (p(n_2+1)) \varphi(n_1) \  $$
where $p\in [0,2\pi)$ and  $\varphi=({h}^{(0)}_{\mathrm{And}}-E)^{-1}\delta_0$ with $E$ the eigenvalue and  ${h}^{(0)}_{\mathrm{And}}$ as in \eqref{eq:AM}.  For $\psi$ to satisfy the eigenfunction equation, $p$ and $E$ must be related by
\begin{equation}\label{eq:dispersion} -2\gamma \cos(p) \ = \ \Sigma(E)  \ , \end{equation}
where $\Sigma(E)$ denotes the Weyl function of the one-dimensional Anderson model on the half-line, i.e, $\Sigma(E) \ := \ -\frac{1}{\langle \delta_0, (\mathcal{h}^{(0)}_{\mathrm{And}} -E -i 0)^{-1} \delta_0 \rangle }$.

Outside of $\sigma(\mathcal{h}_{\mathrm{And}})=[-2,2+\omega_{\mathrm{max}}]$, the Weyl function $\Sigma(E)$ is a smooth, monotonic function of $E$.  There are two smooth maps $p\mapsto E_\pm(p)$ satisfying \eqref{eq:dispersion}, with ranges 
$$J_- \ = \ \left \{ E \le -2 \ : \ \abs{\Sigma(E)} \le 2\gamma \right \} \ \quad \text{and} \quad J_- \ = \ \left \{ E \ge 2+\omega_{\mathrm{max}} \ : \ \abs{\Sigma(E) }\le 2\gamma \right \} \ , $$
respectively.  The transient spectrum of $H_{\mathrm{Vert}}$ is $\sigma_{\mathrm{tac}}=J_- \cup J_+$.  The sets $J_\pm$ are non-deterministic, and for small $\gamma$ one or both may be empty.  The maps $E_\pm(p)$ give dispersion relations for the edge states, which decay exponentially away from $\{n_1=0\}$ by the Combes-Thomas bound (see \cite[Theorem 10.5]{A-W-B}). 

By way of contrast, in the spectrum of $\mathcal{h}_{\mathrm{And}}$ there is no smooth map $p\mapsto E(p)$.  Instead, for each $p$ there is a countable set $S_p$ of energies, dense in $\sigma(\mathcal{h}_{\mathrm{And}})$, at which \eqref{eq:dispersion} holds.  There is no meaningful dispersion relation for these states, since the set $S_p$ varies non-smoothly with $p$. As we show below, the set $\{E\in [-2,2+\omega_{\mathrm{max}}] \ : \ \text{a solution to \eqref{eq:dispersion} exists for some $p$}\} $ is a dense set of positive Lebesgue measure, whose complement is also dense and of positive measure in $[-2,2+\omega_{\mathrm{max}}]$.  However, with probability one, the resulting states still decay exponentially into the bulk due to the localization of the Anderson model Green's function.    

\subsection{Open Questions}
We end this section with some open questions. As a starting point, one may wonder whether the result of Theorem \ref{thmsym} can be improved to show the existence of the limit $\lim_{T\to +\infty}\frac{M^q_{T}\left(H_{\mathrm{Vert}}\right)}{T^q}$ for each $q>0$. More generally, we pose the following question:
\begin{Prob} As $t\to \infty$, does $\frac{1}{t}e^{itH_{\mathrm{Vert}}}X e^{-itH_{\mathrm{Vert}}}$ converge, in the strong sense, to an operator acting in $\ell^2\left(\mathbb{G}_{\mathrm{Vert}}\right)$? 
 \end{Prob}
 
 Existence of the limit $V=\lim_t\frac{1}{t}e^{itH_{\mathrm{Vert}}}X e^{-itH_{\mathrm{Vert}}} $ is called \emph{strong ballistic transport} and is known to occur for limit periodic Schr\"odinger operators, e.g., see \cite{damanik2015}.  The resulting limit $V$ plays the role of a velocity operator, which would typically be related to the derivative $E'(p)$ of the dispersion relation.  For this reason we expect a negative answer to the above question, but it is not obvious how to prove that the limit does not exist.

  It is natural to ask for generalizations of Theorem \ref{thmdiag} on various graphs that extend $\mathbb{G}_{\mathrm{Diag}}$.  Due to the increased number of vertices, these models can be significantly more correlated than the ones covered by Theorem \ref{thmdiag}. For example, one may consider  ``quarter-spaces", for which the underlying graph contains all vertices in $\mathbb{Z}_{\geq 0}\times \mathbb{Z}_{\geq 0}$, see figure 4 below. 
  More precisely, let $\mathbb{G}_{\mathrm{QS}}=\left(\mathcal{V}_{\mathrm{QS}},\mathcal{E}_{\mathrm{QS}}\right)$ where $\mathcal{V}_{\mathrm{QS}}=\mathbb{Z}_{\geq 0}\times \mathbb{Z}_{\geq 0}$ and $(\bm{m},\bm{n})\in \mathcal{E}_{\mathrm{QS}} $ when $\bm{m}$ and $\bm{n}$ are related by one of the following conditions:
$\mathcal{E}^{1}_{\mathrm{QS}}=\{\bm{m}=\bm{n}\pm(1,0), \, \bm{m}, \bm{n}\in \mathcal{V}_{\mathrm{QS}}\}$ or $\mathcal{E}^2_{\mathrm{QS}}=\{\bm{m}=\bm{n}\pm(1,1),\,\bm{m}, \bm{n}\in \mathcal{V}_{\mathrm{QS}}\}$. Let $H_{\mathrm{QS}}=- A_1 - \gamma A_2+V_{\omega}$, with $A_{1,2}$  the adjacency operators of the edge sets $\mathcal{E}^{1,2}_{\mathrm{QS}}$, respectively, and  $V_{\omega}$ as in \eqref{eq:Vomega}.
 \begin{Prob} Is there a value $\gamma_0>0$ for which $\gamma <\gamma_0$ implies at least one of the following?
\begin{enumerate}[label=(\alph*)]
\item  $\sigma\left(H_{\mathrm{QS}}\right)$ is pure point.
\item  $\sup_T M^{q}_{T}\left(H_{\mathrm{QS}} \right ) \ < \ \infty $ with $M^{q}_{T}\left(H_{\mathrm{QS}} \right )$ defined as in \eqref{eq:Mq}.
    \item  \begin{equation}
    \mathbb{E} \left ( \sup_{t\in \mathbb{R}} \left | \langle \delta_{\bm{n}} \ , e^{-itH_{\mathrm{QS}}} \delta_{\bm{m}}\rangle \right | \right ) \ \le \ 
    C_1 e^{-\nu |\bm{m}-\bm{n}|} \ ,
\end{equation} for positive constants $C$ and $\nu$?
\end{enumerate}
 \end{Prob}
 
 In a similar way, one can define ``half-space'' and ``full-space'' versions of $H_{\mathrm{Diag}}$.  For all of these extensions, the proof of Green's function decay given below fails due to the more extensive correlations of the potential.

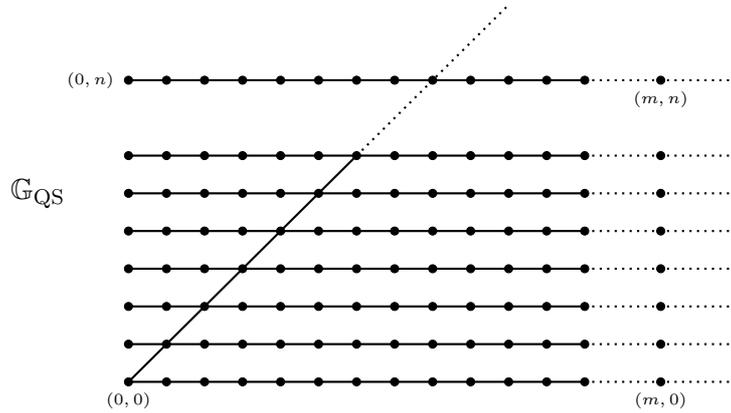
\begin{figure}
\begin{center}
\begin{tikzpicture}[scale=0.5]
\def\width{12};
\def\height{6};
\foreach \m in {0,...,\width}
   \foreach \n in {0,...,\height}
        \draw[fill=black] (\m,\n) circle (3pt);
\foreach \n in {0,...,\height}
    \draw[fill=black] (\number\width +2,\n) circle (3pt);
\foreach \m in {0,...,\width}
    \draw[fill=black] (\m,\number\height +2) circle (3pt);
\draw[fill=black] (\number\width+2,\number\height +2) circle (3pt);
\foreach \n in {0,...,\numexpr\height -1}
    \draw[thick] (\n,\n) -- (\n + 1 , \n + 1);
\foreach \m in {0,...,\numexpr\width -1}
     \foreach \n in {0,...,\height}
         \draw[thick] (\m,\n) -- (\m+1,\n);
\foreach \m in {0,...,\numexpr\width-1}
    \draw[thick] (\m,\number\height +2) -- (\m+1,\number\height +2);
\foreach \n in {0,...,\height} 
     \draw[thick,dotted] (\width,\n) -- (\number\width+2,\n) -- (\number\width+4,\n);
\draw[thick,dotted] (\number\width,\number\height+2) -- (\number\width+2,\number\height + 2) -- (\number\width+4,\number\height + 2) ;
\draw[dotted,thick] (\height,\height)--(\numexpr\height+2,\numexpr\height+2) -- (\numexpr\height+4,\numexpr\height+4);
\node at (-2,{0.5*(\number\height +4)}) {$\mathbb{G}_{\mathrm{QS}}$\,\,\,\,\,\,};
\node at (0,-0.5){\fontsize{6}{9}\selectfont $(0,0)$};
\node at (-1, \number\height+2) {\fontsize{6}{9}\selectfont$(0,n)$};
\node at (\number\width+2,-0.5) {\fontsize{6}{9}\selectfont $(m,0)$};
\node at (\number\width+2,\number\height+2-0.5) {\fontsize{6}{9}\selectfont $(m,n)$};
\end{tikzpicture}
\end{center}
\caption{The diagonal model in the quarter-space }
\label{figure 4}
\end{figure}

\section{Analysis of $H_{\mathrm{Vert}}$ -- proofs of Theorems \ref{thmsym}, \ref{booleleb} and \ref{transrecthm}}\label{Sec:proofofthmsym}
 
\subsection{Absolute continuity of $\mu_{\bm{0}}$}
 \label{acsection}
In this section, we take the first step toward proving the three theorems on $H_{\mathrm{Vert}}$:
\begin{lem}\label{lem:mudelta00ac} The spectral measure $\mu_{\bm{0}}$ for $H_{\mathrm{Vert}}$ associated to $\delta_{\bm{0}}$ is absolutely continuous.
\end{lem}

To prove Lemma \ref{lem:mudelta00ac} a useful tool is the Green's function, defined for $z\in \mathbb{C}\setminus \mathbb{R}$ by \eqref{Greendfn}. In particular,
\begin{equation}
G_{\mathrm{Vert}}\left(\bm{0},\bm{0};z\right)=\langle\delta_{\bm{0}},(H_{\mathrm{Vert}}-z)^{-1}\delta_{\bm{0}}\rangle.
\end{equation}
Its boundary values
$G_{\mathrm{Vert}}\left(\bm{0},\bm{0};E+i0\right):=\lim_{\varepsilon\to 0^{+}}G_{\mathrm{Vert}}\left(\bm{0},\bm{0};E+i\varepsilon\right)$ are well defined for Lebesgue almost every $E\in \mathbb{R}$ by a theorem of de la Vall\'e-Poussin, see \cite[Proposition B.3]{A-W-B} and references therein. Moreover, by \cite[Proposition B.4]{A-W-B}, the singular component of $\mu_{\bm{0}}$ with respect to Lebesgue measure is supported on the set 
\begin{equation}\label{singularsupport}\{E\in \mathbb{R}:\,\,\mathrm{Im}\, G_{\mathrm{Vert}}\left(\bm{0},\bm{0};E+i0\right)=\infty\}.
\end{equation}
We will prove that $\mu_{\bm{0}}$ is absolutely continuous by showing that the above set is empty.

Let $\mathbb{G}^{+}_{\mathrm{Vert}}$ be the component of $\mathbb{G}_{\mathrm{Vert}}$ which contains $(0,1)$ and is obtained from $\mathbb{G}_{\mathrm{Vert}}$ by deleting the edge connecting $(0,0)$ to $(0,1)$.  Denote by $H^{+}_{\mathrm{Vert}}$ the restriction of $H_{\mathrm{Vert}}$ to $\ell^2\left( \mathbb{G}^{+}_{\mathrm{Vert}}\right)$ and define $U:\ell^2\left( \mathbb{G}^{+}_{\mathrm{Vert}}\right)\rightarrow \ell^2\left( \mathbb{G}_{\mathrm{Vert}}\right)$ by $\left(U\psi\right)(m,n)=\psi(m,n-1).$ 
It follows from the geometric resolvent identity that 
\begin{equation}\label{recursive0}
G_{\mathrm{Vert}}\left(\bm{0},\bm{0};z\right) \ = \  \frac{G_{\mathrm{And}}^{(0)}(0,0;z)}{1 - \gamma^2G_{\mathrm{And}}^{(0)}(0,0;z) G^{+}_{\mathrm{Vert}}\left((0,1),(0,1);z\right)} \ ,
\end{equation}
where $G_{\mathrm{And}}^{(0)}(0,0;z)=\langle\delta_0,(\mathcal{h}^{(0)}_{\mathrm{And}}-z)^{-1} \delta_0\rangle$ is the Green's function of the Anderson model \eqref{eq:AM} and $G_{\mathrm{Vert}}^+$ denotes the Green's function of $H_\mathrm{Vert}^+$. However, since the random potential depends only on the first coordinate of the position, we have 
\begin{equation}\label{eq:symmetry}U^{\ast}H_{\mathrm{Vert}}U=H^{+}_{\mathrm{Vert}}.\end{equation}
Therefore $G^{+}_{\mathrm{Vert}}\left((0,1),(0,1);z\right)=G_{\mathrm{Vert}}\left(\bm{0},\bm{0};z\right)$ from which it follows, using \eqref{recursive0}, that
\begin{equation}\label{formula}
2\gamma^2 G_{\mathrm{Vert}}\left(\bm{0},\bm{0};z\right)=w-w\left(1-\frac{4\gamma^2}{w^2} \right)^{1/2} \ , \quad  \text{with } \quad  w= \tfrac{1}{G^{(0)}_{\mathrm{And}}(0,0;z)} \ .
\end{equation}
In \eqref{formula} we take the principal branch of the square root since $\mathrm{Im}\,G_{\mathrm{Vert}}\left(\bm{0},\bm{0};z\right)>0$ whenever $\mathrm{Im}z>0$ and, on the other hand, $\mathrm{Im} \, w = \mathrm{Im}\,\tfrac{1}{G^{(0)}_{\mathrm{And}}(0,0;z)}<0$.

The function $F(w)=w\left ( 1-\frac{4\gamma^2}{w^2} \right )^{1/2} - w$ satisfies $|F(w)|\le 2\gamma$ in the upper half plane $\{\Im w \ge 0\}$. 
Therefore, letting $\mathrm{Im}z\to 0$ in (\ref{formula}), we see that 
$|G_{\mathrm{Vert}}\left(\bm{0},\bm{0};E+i0\right)|\le \frac{1}{\gamma}$ for all $E\in \mathbb{R}$.  
In particular, the set in \eqref{singularsupport} is empty and the spectral measure $\mu_{\bm{0}}$ is purely absolutely continuous with a bounded density; see \cite[Appendix B, Proposition
B.4]{A-W-B}.\footnote{It is worth noting that this argument does not depend on the fact that we take the Anderson model on the horizontal layers.  Indeed, the same argument shows that if $H$ is of the form
$$H\psi(n_1,n_2)= \gamma I[n_1=0] \left [\psi(0,n_2+1) + I[n_2\ge 1] \psi(0,n_2-1) \right ] + \mathcal{h}\otimes I \psi(n_1,n_2)$$
with $\mathcal{h}$ any self-adjoint operator on $\ell^2(\mathbb{Z}_{\ge 0})$, then the Green's function of $H$ is bounded by $\frac{1}{\gamma}$ and the spectral measure $\mu_{\bm{0}}$ for $H$ is absolutely continuous.} This completes the proof of Lemma \ref{lem:mudelta00ac}.

\subsection{Floquet Theory and horizontal localization for $H_{\mathrm{Vert}}$}\label{floquetsection}
The vertical symmetry \eqref{eq:symmetry} of the graph $\mathbb{G}_{\mathrm{Vert}}$ and the definition of the operator $H_{\mathrm{Vert}}$ suggest the use of
a Fourier transform to help study the dynamics $e^{-itH_{\mathrm{Vert}}}$. Given $\psi \in \ell^1(\mathbb{G}_{\mathrm{Vert}})$, let
\begin{equation}\label{Fourierdfn}
\left(\mathcal{F}\psi\right)(n_1,p):=\sqrt{\frac{2}{\pi}}\sum^{\infty}_{n_2= 0}\psi(n_1,n_2)\sin\left(p(n_2+1)\right).
\end{equation}
Initially defined for $\psi\in \ell^1(\mathbb{G}_{\mathrm{Vert}})$, $\mathcal{F}$ may be extended to $\ell^2(\mathbb{G}_{\mathrm{Vert}})$ since $\big{\{}\sqrt{\frac{2}{\pi}}\sin(mp):\,\,m\in \mathbb{N}\big{\}}$ is a complete orthonormal system in $L^2[0,\pi]$. One shows that 
\begin{equation}\label{inversefourier}\mathcal{F^{-1}}(g)(n_1,n_2)=\sqrt{\frac{2}{\pi}}\int^{\pi}_{0}g(n_1,p)\sin\left(p(n_2+1)\right)\,dp
\end{equation}
is a unitary map from $\mathbb{L}^2\left([0,\pi];\ell^2\left(\mathbb{Z}_{\geq 0}\right)\right)$ onto $\ell^2(\mathbb{G}_{\mathrm{Vert}})$ and satisfies $\mathcal{F}^{-1}\mathcal{F}\psi = \psi$ for $\psi\in \ell^1(\mathbb{G}_{\mathrm{Vert}})$. Therefore $\mathcal{F}$ given by (\ref{Fourierdfn}) may be extended to a unitary map
\begin{equation*}
\mathcal{F}:\ell^2\left(\mathbb{G}_{\mathrm{Vert}}\right)\rightarrow 
\mathbb{L}^2\left([0,\pi];\ell^2\left(\mathbb{Z}_{\geq 0}\right)\right)
\end{equation*} with inverse given by (\ref{inversefourier}).
For simplicity of notation, we let $\widehat{\psi}(n_1,p)=\mathcal{F}\left(\psi\right)(n_1,p)$.
It is immediate from the above argument that the following version of Plancherel's identity holds
\begin{equation}\label{Plancherel}
\langle \widehat{\varphi},\widehat{\psi} \rangle_{\mathbb{L}^2\left([0,\pi];\ell^2\left(\mathbb{Z}_{\geq 0}\right)\right)}=\langle \varphi,\psi
\rangle_{\ell^2\left(\mathbb{G}_{\mathrm{Vert}}\right)}.
\end{equation}

From the definition of $H_{\mathrm{Vert}}$ one readily sees that
\begin{equation}\label{Fourier}
\widehat{{H}_{\mathrm{Vert}}\psi}(m,p) \ = \ \mathcal{h}^{(0)}_{\mathrm{And}} \widehat{\psi}(m,p) - 2\gamma\cos p \delta_{m=0} \widehat{\psi}(m,p) \ ,
\end{equation}
where the Anderson model $\mathcal{h}_{\mathrm{And}}^{(0)}$ (see \eqref{eq:AM}) acts on the first coordinate $m$, namely
\begin{equation}\label{eq:half-line And}
\mathcal{h}_{\mathrm{And}}^{(0)}\widehat{\psi}(m,p) \ = \ -\widehat{\psi}(m+1,p) - I[m\ge 1] \widehat{\psi}(m-1,p) +\omega(m)\widehat{\psi}(m,p) \ .    
\end{equation}
Equation (\ref{Fourier}) shows that $H_{\mathrm{Vert}}$ is unitarily equivalent to the direct integral $\int^\oplus_{[0,\pi]} \mathcal{h}_p$ on $$\mathbb{L}^2\left([0,\pi];\ell^2\left(\mathbb{Z}_{\geq 0}\right)\right) \cong \int_{[0,\pi]}^\oplus \ell^2(\mathbb{Z}_{\geq 0}) \ , $$
with the operators $\mathcal{h}_p$ on each fiber given by a rank-one perturbation of the Anderson model:
\begin{equation}\label{defhp}
\mathcal{h}_{p} \varphi \ = \ \mathcal{h}_{\mathrm{And}}^{(0)}\varphi - 2\gamma \cos p \langle \delta_0,\varphi \rangle \delta_0 \ ,
\end{equation}
for $\varphi\in\ell^2\left(\mathbb{Z}_{\geq 0}\right)$.
The following result on dynamical localization for $\mathcal{h}_{p}$ will be technically useful.

\begin{lem}\label{horloc}
Given $s\in (0,1)$ there exist positive constants $C_{\mathrm{And}}(s)$ and $\mu_{\mathrm{And}}=\mu_{\mathrm{And}}(s)$ such that, for all $m,n \in \mathbb{Z}_{\geq 0}$,
\begin{equation}\label{dynlocand}
\mathbb{E}\left(\sup_{|f|\leq 1}|\langle \delta_m,f\left(\mathcal{h}_{p}\right)\delta_n\rangle|\right)\ \leq \ 
Ae^{-\frac{\mu_{\mathrm{And}}}{2-s}|m-n|} \ ,
\end{equation}
with $A = \left((2\omega_{\max}+2\gamma)^{s}\|\rho\|_{\infty}\left(4+4\gamma+\omega_{\max}\right)C_{\mathrm{And}}(s)\right)^{\frac{1}{2-s}}$ and the supremum taken over all Borel measurable functions bounded by one.
\end{lem}
\noindent \textit{Remark}: This result follows easily from known results for the one-dimensional Anderson model, e.g., see \cite[Chapter 12]{A-W-B}, via rank-one perturbation formulas.  For completeness, we give a sketch of the proof in Appendix B. We note that the constants $C_{\mathrm{And}}$ and $\mu_{\mathrm{And}}$  are as in \eqref{eq:fmdecay}.

The localization  for $\mathcal{h}_p$ described in Lemma \ref{horloc} can  immediately be translated into a strong form of horizontal localization
for $H_{\mathrm{Vert}}$. For each $m_1\in \mathbb{Z}_{\ge 0}$ let $P_{m_1}$ denote the orthogonal projection of $\ell^2(\mathbb{G}_{\mathrm{Vert}})$ onto $\mathrm{Span}\{\delta_{(m_1,m_2)}\,\, |\,\ m_2\in \mathbb{Z}_{\geq 0}\}$.  We have the following

\begin{lem}\label{horizontalprojbound} For $s\in (0,1)$ let $C_{\mathrm{And}}(s)$, $\mu_{\mathrm{And}}(s)$ and $A$ be as in Lemma \ref{horloc}. Then for all $m_1,n_1\in \mathbb{Z}_{\ge 0}$ and $\varphi\in \ell^2(\mathbb{G}_{\mathrm{Vert}})$ we have
\begin{equation}\label{dynlocHvert}
\mathbb{E}\left(\sup_{|f|\leq 1} \norm{P_{m_1}f\left(H_{\mathrm{Vert}}\right)P_{n_1}\varphi}^2 \right) \ \leq \  A \|{\varphi}\|^2_{2}e^{-\frac{\mu_{\mathrm{And}}}{(2-s)}|m_1-n_1|}.  
    \end{equation}
    where the supremum is taken over all Borel measurable functions bounded by one.
\end{lem}

\begin{proof} By Plancherel's identity (\ref{Plancherel}) and \eqref{Fourier}, we have 
$$\langle \psi, P_{m_1} f\left(H_{\mathrm{Vert}}\right) P_{n_1}\varphi\rangle 
\ = \ \int^{\pi}_{0}\langle \widehat{P_{m_1}\psi},f(\mathcal{h}_{p})\widehat{P_{n_1}\varphi}\rangle\,dp
\ =\ \int^{\pi}_{0}\overline{{\widehat\psi}(m_1,p)}{\widehat \varphi}(n_1,p)\langle \delta_{m_1},f(\mathcal{h}_{p})\delta_{n_1}\rangle\,dp \ ,
$$
for any $\psi, \varphi\in \ell^2(\mathbb{G}_{\mathrm{Vert}})$.  
Taking absolute values and the supremum over $\psi$ with $\norm{\psi}_2 \le 1$ yields, by Cauchy-Schwarz and \eqref{Plancherel},
$$ \norm{P_{m_1} f\left(H_{\mathrm{Vert}}\right) P_{n_1}\varphi}^2 \ \le  \ \int_0^\pi \abs{\langle \delta_{m_1},f(\mathcal{h}_{p})\delta_{n_1}\rangle}^2 
\abs{\widehat \varphi(n_1,p)}^2 dp  .$$
We find after taking the expectation that
$$\mathbb{E} \left ( \sup_{|f|\le 1} \norm{P_{m_1} f\left(H_{\mathrm{Vert}}\right) P_{n_1}\varphi}^2  \right ) \ \le \ 
 \int_0^\pi  \mathbb{E} \left (  \sup_{|f|\le 1} \abs{\langle \delta_{m_1},f(\mathcal{h}_{p})\delta_{n_1}\rangle} \right ) 
\abs{\widehat \varphi(n_1,p)}^2 dp \ ,$$
where we noted that $\abs{\langle \delta_{m_1},f(\mathcal{h}_{p})\delta_{n_1}\rangle} \le 1$ for $|f|\le 1$.
Using Plancherel's identity \eqref{Plancherel} one more time, the result now follows from Lemma \ref{horloc}.
\end{proof}
    
\begin{cor} For each $q$, we have $\sup_T M^{q}_{T}\left(H_{\mathrm{Vert}},\abs{X_1}\right) < \infty$.
\end{cor}
    \begin{proof}
     Observe that \begin{align*}
         M^{q}_{T}\left(H_{\mathrm{Vert}},\abs{X_1}\right)\ &= \ \frac{2}{T}\int_0^\infty e^{-\frac{2t}{T}}  \sum_{\bm{n}\in {\mathbb{G}}_{\mathrm{Vert}}} |n_1|^q\mathbb{E}\left(|\langle \delta_{\bm{n}},e^{-itH_{\mathrm{Vert}}}\delta_{\bm{0}}\rangle |^2\right)  dt  \\
         & = \ \frac{2}{T}\int_0^\infty e^{-\frac{2t}{T}} \sum^{\infty}_{n_1=0}|n_1|^q \mathbb{E} \left ( \|P_{n_1}e^{-itH_{\mathrm{Vert}}}P_0\delta_0\|^2\right )
         dt \\
         &\le \ \sum^{\infty}_{n_1=0}|n_1|^q\mathbb{E} \left ( \sup_{t} \|P_{n_1}e^{-itH_{\mathrm{Vert}}}P_0\delta_0\|^2\right ) \ < \ \infty , \\
     \end{align*}
     where Lemma \ref{horizontalprojbound} was used in the last step.
    \end{proof}

\subsection{Lower Bound on $M^{q}_{T}\left(H_{\mathrm{Vert}}\right)$: Proof of  Theorem \ref{thmsym}}
\label{proofofthmsym}
For a self-adjoint operator $H$ on $\ell^2(\mathbb{G}_{\mathrm{Vert}})$, the Guarneri bound \cite{Guarnieribound} states that if the spectral measure
$\mu_{\bm{0}}$ is uniformly $\alpha$-H\"older continuous then
\begin{equation}\label{Guarnieri} M^{q}_{T}\left(H\right)\ \geq \ C T^{\frac{\alpha q}{2}}
\end{equation}
holds for some $C>0$.
Recall that a finite Borel measure $\mu$ is said to be \emph{uniformly $\alpha$-H\"older continuous} if there exists a constant $C<\infty$ such that for all intervals $I$ with
$|I|<1$ we have
$\mu(I)\leq C|I|^{\alpha}$; see \cite[Definition 2.2]{A-W-B}. In particular, if the spectral measure $\mu_{\bm{0}}$ for $H$
is purely absolutely continuous with a bounded density, then we have $M^{q}_{T}\left(H\right)\geq C T^{\frac{q}{2}}$. 

To bound $M^q_T(H_{\mathrm{Vert}},|X_2|)$, we will use an adaptation of the proof of the Guarneri bound, incorporating improvements due to the disorder which are specific to our context. To begin, we reproduce the derivation of \eqref{Guarnieri}. The starting point is the following estimate on averaged quantum dynamics in an abstract context.
\begin{thm}[Strichartz-Last] \label{Strichartz-Last} Let $H$ be a self-adjoint operator on a Hilbert space $\mathcal{H}$ and assume the
spectral measure of $H$ with respect to $\psi$ is uniformly $\alpha$-H\"older continuous for some $\alpha\in [0,1]$. Then, there exists a
constant $C_{\psi}<\infty$ such that for all $\phi \in \mathcal{H}$ and all $T>0$
\begin{equation}
\frac{1}{T}\int^{T}_{0}|\langle\phi,e^{-itH}\psi\rangle|^2\leq \frac{C_{\psi}\|\phi\|^2}{T^{\alpha}}.
\end{equation}
\end{thm}

This result, which may be found in \cite[Theorem 2.3]{A-W-B}, can be used to prove the Guarnieri bound (\ref{Guarnieri}) as follows. Suppose that the spectral measure $\mu_{\bm{0}}$ for $H$  
is uniformly $\alpha$-H\"older continuous. Writing $M^{q}_{T}\left(H\right)=\frac{2}{T}\sum_{\bm{n}}|\bm{n}|^{q}\int^{\infty}_{0}e^{\frac{-2t}{T}}|\langle \delta_{\bm{n}},e^{-itH_{\omega}}\delta_{\bm{0}}\rangle|^2\,dt$, we obtain
\begin{align*}
M^q_T(H) \ &\geq \ 2e^{-2}N^q
\sum_{|{\bm{n}}|>N}\frac{1}{T} \int^{T}_{0}|\langle \delta_{\bm{n}},e^{-itH}\delta_{\bm{0}}\rangle|^2\,dt\\ &= \ 2e^{-2}N^q\left(1- \sum_{|{\bm{n}}|\leq N}\frac{1}{T}\int^{T}_{0}|\langle \delta_{\bm{n}},e^{-itH}\delta_{\bm{0}}\rangle|^2\,dt \right) \ . 
\end{align*}
Applying the Strichartz-Last Theorem to each term in the sum on the right hand side, we see that 
$$M^q_T(H) \ \geq \ 2e^{-2}N^q\left(1-\frac{N(N+1)}{2}\frac{C_{\delta_0}}{T^{\alpha}}\right)\ 
 \geq \ CT^{\frac{q \alpha}{2}} \ , $$
 where $\frac{N(N+1)}{2}$ counts the number of sites in $\mathbb{G}_{\mathrm{Vert}}$ with $|{\bm{n}}|\le N$ and in the last step we chose $N^2$ comparable to $T^{\alpha}$ and adjusted the constant accordingly.

For $H=H_{\mathrm{Vert}}$, we now follow the proof of \eqref{Guarnieri} to estimate $\mathbb{E}(M_T^q(H,|X_2|))$. First, we have 
\begin{equation}\mathbb{E}\left ( M_T^q(H,|X_2|) \right ) \ \ge \ 2 e^{-2} N^q \left ( 1 - \sum_{n_1=0}^\infty \sum_{n_2=0}^N \frac{1}{T} \int_0^T 
\mathbb{E}\left (|\langle \delta_{(n_1,n_2)} ,e^{-itH}\delta_{\bm{0}}\rangle|^2 \right ) \,dt \right ) \ .\label{MTqX2bound}\end{equation}
Let $K\ge 0$ be a sufficiently large integer to be specified below. Applying the localization bound of Lemma \ref{horizontalprojbound} for $n_1\ge K$ and the Strichartz-Last Theorem \ref{Strichartz-Last} for $n_1<K$, we find that
$$
\mathbb{E}\left ( M_T^q(H,|X_2|) \right ) \ \ge  \ 2e^{-2} N^q \left (1-  C K\frac{N+1}{T}  -  \frac{A_s}{1-e^{-\frac{\mu_{\mathrm{And}}}{2-s}}}   e^{-\frac{\mu_{\mathrm{And}}}{2-s}K}  \right )
$$
for $s\in (0,1)$.  Choosing $K$ sufficiently large and then taking $N\propto \frac{T}{K}$ (for $T$ sufficiently large) 
yields the desired bound, equation \eqref{eqsym}, and completes the proof of Theorem \ref{thmsym}.

\subsection{A formula for spectral measures: proof of Theorem \ref{booleleb}}\label{sec:booleleb}
The proof of Theorem \ref{booleleb} is based on an exact formula for spectral measures of $H_{\mathrm{Vert}}$ which is also crucial to the analysis of the transient and recurrent components in the next section.  Let 
\begin{equation}\label{eq:Weyl}
\Sigma(E) \ := \ -\frac{1}{\langle \delta_0, (\mathcal{h}_{\mathrm{And}}^{(0)} -E -i 0)^{-1} \delta_0 \rangle }
\end{equation}
denote the \emph{Weyl} function for the Anderson model \eqref{eq:AM}, where the boundary values exists for Lebesgue almost every $E \in \mathbb{R}$ (see, e.g., \cite[Theorem 5.9.1]{Simon-Book}).

\begin{lem}\label{Lem:formulameasure} Denote by 
$\mu_{n_2,n'_2}$ the spectral measure of $H_{\mathrm{Vert}}$ associated to $\delta_{(0,n_2)}$ and $\delta_{(0,n'_2)}$. Then, for $n_2,n'_2\in \mathbb{Z}_{\geq0}$ we have that
\begin{equation}\label{formulaspectralmeasure}
d\mu_{n_2,n'_2}(E)=\frac{1}{\pi\gamma}\sqrt{\left [ 1-\frac{\Sigma(E)^2}{4\gamma^2} \right ]_+}U_{n_2}\left(\frac{\Sigma(E)}{2\gamma}\right) U_{n_2'}\left(\frac{\Sigma(E)}{2\gamma}\right) \,dE,
 \end{equation}
 where $[x]_+:=\max(x,0)$ denotes the positive part and $U_m(z)$, $m=0,1,2,\ldots$, denote the Chebyshev polynomials of second kind.
\end{lem}
\begin{proof}
Let $L\in \mathbb{N}$ be given and let $H^{L}_{\mathrm{Vert}}$ denote the restriction of $H_{\mathrm{Vert}}$ to $\ell^{2}\left(\{0,...,L\}\times \mathbb{Z}_{\geq 0}\right)$. Since
$$
 \widehat{\delta}_{(0,n_2)}(m_1,p)\ =\ \sqrt{\frac{2}{\pi}} \, \delta_{0}(m_1)\sin\left(p(n_2+1) \right),
$$
it follows from Plancherel's identity (\ref{Plancherel}) that if $f$ is, for example, a bounded continuous function, then
 \begin{align*} \langle \delta_{(0,n_2)},f\left(H^{L}_{\mathrm{Vert}}\right)\delta_{(0,n'_2)}\rangle \ &= \ \frac{2}{\pi}\int^{\pi}_{0}\langle
 \delta_0,f\left(\mathcal{h}^L_p\right)\delta_0\rangle\sin\left((n_2+1) p\right)\sin\left((n'_2+1) p\right)\,dp.\\
 &= \ \frac{2}{\pi}\int^{\pi}_{0}\int^{\infty}_{-\infty}\left(f(E)d\mu^{p,L}_{0}(E)\right)\sin\left((n_2+1) p\right)\sin\left((n'_2+1) p\right)\,dp.
 \end{align*}
 where $\mu^{p,L}_{0}$ is the spectral measure for $\mathcal{h}^L_p$ associated to $\delta_0$ and $\mathcal{h}^L_p$ is the restriction of $\mathcal{h}_p$  to $\ell^2(\{0,\ldots,L\})$ (see eq.\ \eqref{defhp}). Recalling that 
 $U_{m}(\cos p)=\frac{\sin\left((m+1)p\right)}{\sin p}$,
 we obtain
 \begin{equation}\label{spectralmeas:pform}
     \langle \delta_{(0,n_2)},f\left(H^{L}_{\mathrm{Vert}}\right)\delta_{(0,n'_2)}\rangle=\frac{2}{\pi}\int^{\pi}_{0}\int^{\infty}_{-\infty}\left(f(E)d\mu^{p,L}_{0}(E)\right)U_{n_2}(\cos p)U_{n'_2}(\cos p)\sin^2p\,dp
 \end{equation}

 Let $\mathcal{h}^{(0,L)}_{\mathrm{And}}$ denote the restriction of the Anderson model $\mathcal{h}_{\mathrm{And}}^{(0)}$ to $\ell^2\left(\{0,...,L\}\right)$.  Since $\mathcal{h}^L_p$ is a rank-one perturbation of $\mathcal{h}^{(0,L)}_{\mathrm{And}}$, it follows from  \cite[Theorem 5.3]{A-W-B} that
 \begin{equation}\label{distributional}d\mu^{p,L}_{0}(E)=\delta\left(\Sigma_{L}(E)+2\gamma \cos p\right)\,dE
 \end{equation}
where $\frac{1}{\Sigma_{L}(E)}=- \langle \delta_0,\left (\mathcal{h}^{(0,L)}_{\mathrm{And}}-E \right )^{-1} \delta_0 \rangle$. As explained in \cite{A-W-B}, the distributional identity (\ref{distributional}) is equivalent to stating that $\mu^{p,L}_{0}$ is supported on the set $\{E\in \mathbb{R}:\,\,\Sigma_{L}(E)=-2\gamma \cos p\}$, assigning to each point is this set the mass $\mu^{p,L}_{0}(\{E\})=\frac{1}{\Sigma'_{L}(E)}$. Given $E\in \mathbb{R}$ with  $\Sigma_{L}(E)\in(-2\gamma,2\gamma)$, let $q=q(E)$ denote the unique value in $(0,\pi)$ for which $-2\gamma \cos (q)=\Sigma_{L}(E).$ We have that
 \begin{equation}\label{eq:changeofvar} 
 \delta\left(\Sigma_{L}(E)+2\gamma \cos p\right)dp=\frac{1}{2\gamma \sin(q)}\delta(q-p)dp,
 \end{equation}
 see equation (5.11) in \cite{A-W-B}.
 
In particular, letting $P_{n_2,n'_2}(z)=U_{n_2}(z)U_{n'_2}(z)$ it follows from Fubini's theorem along with equations \eqref{spectralmeas:pform}, \eqref{distributional} and \eqref{eq:changeofvar} that
 \begin{equation}
     \langle
 \delta_{(0,n_2)},f\left(H^L_{\mathrm{Vert}}\right)\delta_{(0,n'_2)}\rangle \ =\ \frac{1}{\pi\gamma}\int_{-\infty}^\infty f(E)\sqrt{\left[1-\frac{\Sigma_{L}(E)^2}{4\gamma^2}\right ]_+}P_{n_2,n_2'}\left(\frac{\Sigma_{L}(E)}{2\gamma}\right) \,dE \ ,\label{eq:lastofformulaspectralmeasure}
 \end{equation}
 where we have used that $\int^{\pi}_{0}\delta(p-q)\,dp=1$ holds for $q\in (0,\pi)$.
Letting $L\to \infty$, eq.\ \eqref{formulaspectralmeasure} follows from the dominated convergence theorem.\qedhere
\end{proof}

A first consequence of the prior lemma is the cyclicity of $\delta_{\bm{0}}$:
\begin{lem}\label{cyclic}
$\delta_{\bm{0}}$ is a cyclic vector for $H_{\mathrm{Vert}}.$
\end{lem}

\begin{proof}
Denote by ${\mathcal{H}}_{\bm{0}}$ the cyclic subspace of $H_{\mathrm{Vert}}$ associated to $\delta_{\bm{0}}$. Let $n_2\in \bbZp$ and let $P_{{\mathcal{H}}_{0,0}}\delta_{(0,n_2)}$ be the orthogonal projection of $\delta_{(0,n_2)}$ onto ${\mathcal{H}}_{\bm{0}}$.  By lemma \ref{Lem:formulameasure}, we have  $$d\mu_{0,n_2}(E)\ =\ U_{n_2}\left(\frac{\Sigma(E)}{2\gamma}\right)\,d\mu_{\bm{0}}(E) \ ,$$ where $d\mu_{0,n_2}$ is the spectral measure associated to the pair of vectors $(\delta_{(n_2,0)},\delta_{\bm{0}})$ and 
\begin{equation}\label{eq:lmu00}d\mu_{\bm{0}}(E) \ = \ \frac{1}{\pi\gamma}\sqrt{\left[1-\frac{\Sigma(E)^2}{4\gamma^2}\right]_+} \, d E
\end{equation}
is the spectral measure associated to $\delta_{\bm{0}}.$    It follows from the spectral theorem that
\begin{equation*}
\|P_{{\mathcal{H}}_{\bm{0}}}\delta_{(0,n_2)}\|^2 \ = \ \int^{\infty}_{-\infty}U_{n_2}^{2}\left(\frac{\Sigma(E)}{2\gamma}\right)d{\mu}_{\bm{0}}(E) 
\ .
\end{equation*}
Applying Lemma \ref{Lem:formulameasure} a second time with $n_2=n'_2$, we conclude that
$$ \|P_{{\mathcal{H}}_{\bm{0}}}\delta_{(0,n_2)}\|^2 \ = \ \frac{1}{\pi \gamma} \int^{\infty}_{-\infty}U_{n_2}^{2}\left(\frac{\Sigma(E)}{2\gamma}\right) \sqrt{\left[1-\frac{\Sigma(E)^2}{4\gamma^2}\right ]_+} \ = \ \left \| \delta_{(0,n_2)} \right \|^2 \ = \ 1 \ . $$
Thus $\delta_{(0,n_2)}\in {\mathcal{H}}_{\bm{0}}$ for each $n_2\in \mathbb{Z}_{\geq 0}$. It readily follows from the definition of $H_{\mathrm{Vert}}$ that $\ell^2\left(\mathbb{G}_{\mathrm{Vert}}\right)=\overline{\mathrm{Span}\{H^{(j)}_{\mathrm{Vert}}\delta_{(0,n_2)}:\,\,j,n_2\in \mathbb{Z}_{\geq 0}\}}$. Therefore $\delta_{\bm{0}}$ is a cyclic vector for $H_{\mathrm{Vert}}.$ 
\end{proof}

We are now ready to complete the proof of Theorem \ref{booleleb}. Since $\delta_{\bm{0}}$ is cyclic and $\mu_{\bm{0}}$ is absolutely continuous, it follows that $H_{\mathrm{Vert}}$ has simple, purely absolutely continuous spectrum.  By Lemma \ref{Lem:formulameasure}, $\mu_{\bm{0}}$ is supported on $\{ E \ : \ |\Sigma(E)| <2\gamma\}$ and its density is  bounded by $\frac{1}{\pi\gamma}$ (see eq.\ \eqref{eq:lmu00}).  To compute the Lebesgue measure of the support, we use the following
\begin{lem}\label{lem:recursiveanderson} The Weyl function $\Sigma(E)$ satisfies
\begin{equation}\label{recursiveanderson}
\Sigma(E)=E-\omega(0)+\langle \delta_1 , (\mathcal{h}_{\mathrm{And}}^{(1)}-E-i 0 )^{-1} \delta_1 \rangle \ ,
\end{equation}
with $\mathcal{h}_{\mathrm{And}}^{(1)}$ the Anderson model \eqref{eq:AM} on $\ell^2(\mathbb{Z}_{\ge 1})$.
\end{lem}
\begin{proof} Letting
$\varphi=(\mathcal{h}_{\mathrm{And}}^{(0)}-z)^{-1}\delta_0$, we have by definition of $\mathcal{h}_{\mathrm{And}}^{(0)}$ that
\begin{equation*}
-\varphi(1)+\left(\omega(0)-z\right)\varphi(0) \ = \ 1 \ .
\end{equation*}
Since, by the geometric resolvent identity,
\begin{equation*}
\varphi(1) \ = \ \langle \delta_1 , (\mathcal{h}_{\mathrm{And}}^{(0)}-E-i 0 )^{-1} \delta_0 \rangle \ =\ \langle \delta_1 , (\mathcal{h}_{\mathrm{And}}^{(1)}-E-i 0 )^{-1} \delta_1 \rangle \varphi(0) \ ,
\end{equation*}
eq.\ \eqref{recursiveanderson} follows by taking $z=E+i\varepsilon$ and $\varepsilon \to 0^+$.
\end{proof}
Returning to the proof of Theorem \ref{booleleb}, we see that 
$$\left | \left \{ E \ : \ |\Sigma(E)|< 2\gamma \right \} \right | \ = \ \left | \left \{ E \ : \ |E-\omega(0)+ F(E+i0) |< 2\gamma \right \} \right | \ , $$
where
$ F(z)  = \langle \delta_1 , (\mathcal{h}_{\mathrm{And}}^{(1)}-z )^{-1} \delta_1 \rangle  = \int_{\mathbb{R}} \frac{1}{u-z} d \mu_1(u)$,
with $\mu_1$ the spectral measure for $\mathcal{h}_{\mathrm{And}}^{(1)}$ associated to $\delta_1$. Almost surely, $\mu_1$ is pure-point, hence purely singular (see, e.g., \cite{J-Z}). Thus, with probability one, we have
$\left | \left \{ E \ : \ |\Sigma(E)|< 2\gamma \right \} \right |  = 4\gamma $,
by Proposition \ref{booletype}.  This completes the proof.

\subsection{Transient and recurrent components: proof of Theorem \ref{transrecthm}}\label{Sec:transient}
We now turn to the study of the transient and recurrent components of the spectrum of $H_{\mathrm{Vert}}$.  We begin by recalling some measure theoretic topology from \cite{A-S}.  An \emph{event} is an equivalence class $[S]$ of Borel subsets of $\mathbb{R}$ under the relation $S\sim T$ if $S$ and $T$ differ by a Lebesgue measure zero set, i.e., $|S\Delta T| =0$.  The \emph{support} of an absolutely continuous measure $d\mu=f(E)dE$ is the \emph{event} $[\{f>0\}]$. The \emph{essential interior} of an event $[S]$ is the open set
$$ U \ = \ \{ E \ : \ |(E-t,E+t)\cap S| = 2t \quad \text{for some } t >0 \} \ $$
and the \emph{essential frontier} of $[S]$ is the event $[S\setminus U ]$.

Given a self-adjoint operator $H$ on a separable Hilbert space, a \emph{maximal spectral measure for $H$} is a Borel measure $\mu$ such that for any Borel set $A$,  $\mu(A)>0$ if and only if the corresponding spectral projection satisfies $P_A(H)\neq 0$. Every self-adjoint operator $H$  admits a maximal spectral measure $\mu$ and any other spectral measure $\mu_{\varphi}$ is absolutely continuous with respect to $\mu$, see \cite[Lemma 3.16]{Teschl}. The \emph{$H$-event} is the support of the absolutely continuous part of a maximal spectral measure $\mu$ for $H$. If $H$ has a cyclic vector $\psi$, then the $H$-event coincides with the support of the absolutely continuous part of the spectral measure $\mu_\psi$. A key result of Avron and Simon is that the essential interior and essential frontier of the $H$-event determine the transient and recurrent spectrum of $H$:
\begin{thm}[{\cite[Theorem 3.4]{A-S}}]\label{A-Sthm} Let $H$ be a self-adjoint operator on $\mathcal{H}$, a separable Hilbert space. Let $[A]$ be the $H$-event and let $[B]$, $[C]$ be its essential interior and essential frontier.  Then 
$$ \mathcal{H}_{\mathrm{tac}} = E_B \mathcal{H}_{\mathrm{ac}}; \quad \mathcal{H}_{\mathrm{rac}} = E_C \mathcal{H}_{\mathrm{ac}} \ . $$
\end{thm}

For $H_{\mathrm{Vert}}$ we have the following
\begin{lem}\label{lem:Hevent} The $H_{\mathrm{Vert}}$-event is $[S]$ where $ S  :=  \{ E \ : \ |\Sigma(E)|< 2 \gamma \} $. The essential interior of $[S]$ is the open set $S \setminus [-2,2+\omega_{\mathrm{max}}]$
and the essential frontier of $[S]$ is the event $[S\cap [-2,2+\omega_{\mathrm{max}}]].$
\end{lem}

\begin{proof}[Proof of Lemma \ref{lem:Hevent}]That $S$ is a support for $\delta_{\bm{0}}$ follows from Lemma \ref{Lem:formulameasure}; we have already used this fact in our proof of Theorem \ref{booleleb} above. Since, according to Lemma \ref{cyclic}, $\delta_{\bm{0}}$ is cyclic, it follows that $[S]$ is the $H_{\mathrm{Vert}}$-event. Since $\Sigma(E)$ is continuous (analytic, in fact) on $[-2,2+\omega_{\mathrm{max}}]^c $ we see that $S\setminus [-2,2+\omega_{\mathrm{max}}]$ is open, and thus contained in the essential interior of $[S]$.

Since the essential frontier and essential interior are essentially disjoint (see remark following \cite[Proposition 2.2]{A-S}), to complete the proof it suffices to show that $[S\cap [-2,2+\omega_{\mathrm{max}}]]$ is contained in the essential frontier of $[S]$.  Below we show that $S^c$ is essentially dense in $[-2,2+\omega_{\mathrm{max}}]$, i.e., $|(E-t,E+t)\cap S^c| >0$ for any $E\in [-2,2+\omega_{\mathrm{max}}]$ and any $t>0$. It follows that $|(E-t,E+t)\cap S| < 2t$ for any $E\in  [-2,2+\omega_{\mathrm{max}}]$ and $t>0$, so $[S\cap [-2,2+\omega_{\mathrm{max}}]]$ is contained in the essential frontier of $[S]$.
\end{proof}

We recall that a set $T$  is \emph{essentially dense} in an interval $I\subset\mathbb{R}$ if $|J\cap T| >0$ for any interval $J\subset I$ (see \cite[\S 3]{D-S}).  
To complete the proof of Lemma \ref{lem:Hevent} it remains to show the following
\begin{lem}
Both $S$ and $S^c$ are essentially dense in $ [-2,2+\omega_{\mathrm{max}}]$.
\end{lem}
\begin{proof} Let $J \subset [-2,2+\omega_{\mathrm{max}}]$ be an interval.   We must show that $|J\cap S|>0$ and $|J\cap S^c|>0$.  Let $\mu_{0}$ denote the spectral measure of the Anderson model ${\mathcal{h}}_{\mathrm{And}}^{(0)}$ on $\ell^2(\mathbb{Z}_{\ge 0})$ associated to $\delta_0$.  Since $\delta_0$ is cyclic for ${\mathcal{h}}_{\mathrm{And}}^{(0)}$ and $\sigma({\mathcal{h}}_{\mathrm{And}}^{(0)})=[-2,2+\omega_{\mathrm{max}}]$, we conclude that $\mu_{0}(J)>0$.  
Because ${\mathcal{h}}_{\mathrm{And}}^{(0)}$ has pure point spectrum, $\mu_{0}$ is a purely singular measure.  It follows that
$$\lim_{t\to \infty}t\big{|}\{E\in J: |\langle \delta_0 , (\mathcal{h}_{\mathrm{And}}-E-i 0 )^{-1} \delta_0 \rangle |>t\} \big{|} \ = \ \frac{2}{\pi}\mu_{0}(J) \ > \ 0 \ ,
$$
by \cite[Equation (5.4)]{SPZ} (\cite[Theorem 1]{Poltora}, when restated in terms of Borel measures on the real line, would also suffice).  
In particular, choosing t sufficiently large yields
$$\big |S\cap J \big | \ = \ \left | \left \{ E \in J \ : \ \big | \langle \delta_0 , (\mathcal{h}_{\mathrm{And}}-E-i 0 )^{-1} \delta_0 \rangle \big | > \frac{1}{2\gamma} \right \} \right  | >0. $$
Thus, $S$ is essentially dense in $[-2,2+\omega_{\mathrm{max}}]$.
Similarly, letting $\mu^{(1)}_{1}$ denote the spectral measure of the Anderson model $\mathcal{h}^{(1)}_{\mathrm{And}}$ on $\ell^2(\mathbb{Z}_{\ge 1})$ associated to $\delta_1$, we have $\mu_{1}^{(1)}(J) >0$ and
$$\lim_{t\to \infty}t\big{|}\{E\in J: |\langle \delta_1 , (\mathcal{h}_{\mathrm{And}}^+-E-i 0 )^{-1} \delta_1 \rangle|>t\} \big{|} \ = \ \frac{2}{\pi}\mu_{1}^{(1)}(J) \ > \ 0 \ .
$$
By choosing $t$ sufficiently large and using Lemma \ref{lem:recursiveanderson}, we conclude by the triangle inequality that
\begin{equation*}
\big | S^c \cap J \big | \ = \ \big{|}\{E\in J:|\omega(0)-E-\langle \delta_1 , (\mathcal{h}_{\mathrm{And}}^+-E-i 0 )^{-1} \delta_1 \rangle|\ge 2\gamma\}\big{|}>0 \ ,
\end{equation*}
hence $S^c$ is essentially dense in $ [-2,2+\omega_{\mathrm{max}}]$.
\end{proof}

Theorem \ref{transrecthm} follows from the corollary below, which is a direct consequence of Theorem \ref{booleleb}, Theorem \ref{A-Sthm} and Lemma \ref{lem:Hevent}:
\begin{cor}\label{ractacsplit}
Let $P_r \ := \ P_{[-2,2+\omega_{\mathrm{max}}]} (H_{\mathrm{Vert}})$ and $P_t \ := \ P_{[-2,2+\omega_{\mathrm{max}}]^c}(H_{\mathrm{Vert}})$ denote the spectral projections of $H_{\mathrm{Vert}}$ onto $[-2,2+\omega_{\mathrm{max}}]$ and $[-2,2+\omega_{\mathrm{max}}]^c$ respectively. Then  $\mathrm{ran}P_{r} \ = \  \mathcal{H}_{\mathrm{rac}}$ and $\mathrm{ran}P_{t} \ = \ \mathcal{H}_{\mathrm{tac}} \ . $
\end{cor}

Note that the transient spectrum of $H_{\mathrm{Vert}}$, $\sigma_{\mathrm{tac}}=\overline{S\setminus (-2,2+\omega_{\mathrm{max}})}$, is non-deterministic.  Since $\Sigma(E)$ is monotone increasing on each component of the complement of $(-2,2+\omega_{\mathrm{max}})$, $\sigma_{\mathrm{tac}}=J_- \cup J_+$ where $J_-\subset [-2- 2\gamma,-2]$ and $J_+\subset[2+\omega_{\mathrm{max}},2+\omega_{\mathrm{max}}+2\gamma]$ are \emph{intervals} given by 
\begin{equation}\label{Eq:J}J_- \ = \ \left \{ E \le -2 \ : \ \abs{\Sigma(E)} \le 2\gamma \right \} \ \quad \text{and} \quad J_- \ = \ \left \{ E \ge 2+\omega_{\mathrm{max}} \ : \ \abs{\Sigma(E) }\le 2\gamma \right \} \ .
\end{equation}
Depending on the configuration $\omega$ and the hopping $\gamma$, one or both of $J_\pm$ can be empty.  Indeed, by monotonicity of $\Sigma(E)$  we have $J_\pm\neq \emptyset$ if and only if $\gamma > \frac{|\Sigma(E_\pm)|}{2}$ where $E_-=-2$ and $E_+ = 2+\omega_{\mathrm{max}}$.  For completeness we provide a proof that these numbers are non-zero and finite almost surely. Here we commit a slight abuse of notation: since $\Sigma(E)$ is only defined for Lebesgue almost-every $E\in[-2,2+\omega_{\mathrm{max}}]$, the values $\Sigma(E_\pm)$ are understood as the side limits $\Sigma(E_+)=\lim_{E\downarrow E_+} \Sigma(E)$ and $\Sigma(E_-)=\lim_{E\uparrow E_-} \Sigma(E)$ (which are well defined due to monotonicity of $\Sigma(E)$ in $[-2,2+\omega_{\mathrm{max}}]^c$).
\begin{lem}\label{lem:edge} At the edges $E_{-}=-2$ and $E_+=2 + \omega_{\mathrm{max}}$ we have that
\begin{equation}\label{edgebound}
0<|\Sigma(E_{\pm})|<\infty
\end{equation}
almost surely.
\end{lem}
\begin{proof} Assume, for the sake of contradiction, that $ \mathbb{P}\left(\Sigma(E_{\pm })=0\right) > 0$.
Then, Lemma \ref{lem:recursiveanderson}, $$
\mathbb{P}\left(\omega(0)-E_{\pm}-\langle \delta_1 , (\mathcal{h}_{\mathrm{And}}^+-E_{\pm}-i 0 )^{-1} \delta_1 \rangle=0\right)>0 \ . $$
Since $\langle \delta_1 , (\mathcal{h}_{\mathrm{And}}^+-E_\pm-i 0 )^{-1} \delta_1 \rangle $ is independent of $\omega(0)$, the above equation contradicts the fact that the distribution of $\omega(0)$ is purely absolutely continuous. We conclude that
\begin{equation*}
\mathbb{P}\left(\Sigma(E_\pm)=0\right)=0.
\end{equation*}
By a similar argument, one sees that $\mathbb{P}\left(\Sigma(E_{\pm})=\infty\right)=0.$
\end{proof}

\section{Analysis of $H_{\mathrm{Diag}}$ \textemdash \ proof of Theorem \ref{thmdiag} and Corollary \ref{cormain}}\label{sec:Hdiagproofs}
\subsection{Green's function decay}\label{Greendecaysec}
In this section we turn to the analysis of the operator $H_{\mathrm{Diag}}$ on $\ell^2(G_{\mathrm{Diag},\ell})$ for fixed $\ell\ge 0$.  Let $G_{\mathrm{Diag}}({\bm{m}},{\bm{n}};z) = \langle \delta_{\bm{m}} , (H_{\mathrm{Diag}} - z)^{-1} \delta_{{\bm{n}}} \rangle$ for  ${\bm{m}},{\bm{n}}\in \mathbb{G}_{\mathrm{Diag},\ell}$ and $z\in \mathbb{C}\setminus \mathbb{R}$ denote the corresponding Green's function.  A preliminary observation is that the fractional moments of $G_{\mathrm{Diag}}$ are bounded:
\begin{lem}\label{lem:apriori}
There is $C_{\mathrm{AP}}<\infty$ such that for each $s\in (0,1)$, $z\in \mathbb{C}\setminus \mathbb{R}$ and $\bm{m}=(m_1,m_2)$, $\bm{n}=(n_1,n_2)$ in $\mathbb{G}_{\mathrm{Diag},\ell}$ we have
\begin{equation}\label{eq:apriori}
    \mathbb{E} \left ( \left | G_{\mathrm{Diag}}(\bm{m},\bm{n};z) \right |^s \middle | \mathcal{F}_{m_1,n_1}^c \right ) \ \le \ \frac{C_{\mathrm{AP}}^s}{1-s} \ ,
\end{equation}
where $\mathbb{E}\left ( \cdot \middle | \mathcal{F}_{m_1,n_1}^c \right )$ denotes averaging with respect to the variables $\omega(m_1)$ and $\omega(n_1)$.
\end{lem}
\begin{Rems} \begin{enumerate}[label=(\roman*)]
    \item Due to the correlation between potentials at different sites, the rank-one bounds of the original Aizeman-Molchanov method \cite{A-M} do not work here. However, eq.\ \eqref{eq:apriori} is a straightforward consequence of the Hilbert-Schmidt fractional moment bounds developed for continuum Schr\"odinger operators \cite{A-E-N-S-S}. For completeness we give a brief sketch of the proof of Lemma \ref{lem:apriori} in Appendix \ref{sec:apriori}, based on results from \cite{A-E-N-S-S}.
    \item The \emph{a priori} bound $C_{\mathrm{AP}}$ depends on $s$ and the distribution of the random variables, but is independent of $\ell$, $\bm{m}$, $\bm{n}$ and $z$.
    \item The average $\mathbb{E}\left ( \cdot \middle | \mathcal{F}_{m_1,n_1}^c \right )$ is the conditional expectation with respect to the $\sigma$-algebra generated by $\{\omega(r) \ : \ r \neq m_1 \text{ and } r\neq n_1\}$, which explains the notation. 
\end{enumerate}

\end{Rems}
The key estimate we use  below to prove Theorem \ref{thmdiag} is exponential decay of the fractional moments of the Green's function:

\begin{lem}\label{Lemma:decay0n} Let $r\in (0,\frac{1}{3})$.  If $(C_{\mathrm{AP}} \gamma)^{3r}<1-3r$, then there exist constants $C_{\mathrm{Diag}}<\infty$ and $\mu_{\mathrm{Diag}}>0$ such that for all $\bm{m},\bm{n}\in \mathbb{G}_{\mathrm{Diag};\ell}$ and $z\in \mathbb{C}\setminus \mathbb{R}$ we have
\begin{equation}\label{decay0n}
\mathbb{E}\left(|G_{\mathrm{Diag}}\left(\bm{m},\bm{n};z\right)|^{r}\right) \ \leq \ C_{\mathrm{Diag}}e^{-\mu_{\mathrm{Diag}}d_{\mathrm{Diag}}(\bm{m},\bm{n})} \ ,
\end{equation}
where $d_{\mathrm{Diag}}$ denotes the graph distance in $\mathbb{G}_{\mathrm{Diag},\ell}$.
\end{lem}
\begin{Rem}
In terms of the constants  $\mu_{\mathrm{And}}=\mu_{\mathrm{And}}(3r)$, $C_{\mathrm{And}}=C_{\mathrm{And}}(3r)$ appearing in eq.\ \eqref{eq:fmdecay}, with $s=3r$, we have $\mu_{\mathrm{Diag}}=\min\bigl(\frac{1}{3}\mu_{\mathrm{And}},\frac{\alpha}{l+2}\bigr)$ where $\alpha=-\log\bigr (\frac{(C_{\mathrm{AP}} \gamma)^{r}}{(1-3r)^{\frac{1}{3}}} \bigr )$ and $C_{\mathrm{Diag}}=2\max\bigl( {C_{\mathrm{And}}}^{\frac{1}{3}}, {C_{\mathrm{And}}}^{\frac{2}{3}}\frac{C_{\mathrm{AP}}^r}{(1-3r)^{\frac{1}{3}}}e^{2\mu_{\mathrm{Diag}}(2+\ell)}\bigr)$.
\end{Rem}

Lemma \ref{Lemma:decay0n} combines two estimates: 1) localization of the $1D$ Anderson model in the bulk of $\mathbb{G}_{\mathrm{Diag},\ell}$, see eq.\ \eqref{eq:fmdecay}, and 2) decay along the boundary of $\mathbb{G}_{\mathrm{Diag},\ell}$, expressed in the following
\begin{lem}\label{initialgdecay}
Let $s\in (0,1)$ be given. For all $m,m'\in \mathbb{Z}_{\geq 0}$, $j,j'\in \{0,1,...,\ell\}$ and $z\in \mathbb{C}^{+}$ we have
\begin{equation}\mathbb{E}\left(|G_{\mathrm{Diag}}\left((m',m'(\ell+1)+j'),(m,m(\ell+1)+j);z\right)|^{s}\middle | \mathcal{F}_{[m',m]}^c \right) \ \leq \
\frac{C_{\mathrm{AP}}^s}{1-s}\left(\frac{(C_{\mathrm{AP}}\gamma)^s}{1-s}\right)^{|m-m'|}.
\end{equation}
where $\mathbb{E} \left( \cdot \middle |  \mathcal{F}_{[m',m]}^c \right) $ denotes averaging over $\omega(r)$ for each $r$ between $m'$ and $m$.
\end{lem}
\begin{Rem}As above, $\mathbb{E} \left( \cdot \middle |  \mathcal{F}_{[m',m]}^c \right) $ is a conditional expectation, in this case with respect to the $\sigma$-algebra generated by $\{\omega(k) \ : \ k> \max(m,m')\text{ or } k < \min(m,m')\}.$
\end{Rem}
\begin{proof}
Since $H_{\mathrm{Diag}}$ is a real symmetric operator, the Green's function is symmetric: $G(\bm{m},\bm{n};z) = G(\bm{n},\bm{m};z)$.  Thus it suffices to consider $m\ge m'$.
Fix $j,j'\in \{0,1,...,\ell\}$, $z\in \mathbb{C}^{+}$ and let ${\bm{x}}_{m,j}=(m,m(\ell+1)+j)$.
We proceed by induction in $m$. When $m=m'$ the statement reduces to the \emph{a priori} bound of Lemma \ref{lem:apriori}.  Suppose that the desired conclusion holds for some $m\ge m'$. By the geometric resolvent identity, we have the factorization
\begin{equation*}
G\left(\bm{x}_{m',j'},\bm{x}_{m+1,j};z\right)\ = \ \gamma G\left(\bm{x}_{m',j'},\bm{x}_{m,\ell};z\right)G^{+}\left(\bm{x}_{m+1,0},\bm{x}_{m+1,j};z\right).
\end{equation*}
Taking absolute values, raising both sides to the power $s$, and averaging, we find that
\begin{multline*}
\mathbb{E}\left(|G\left(\bm{x}_{m',j'},\bm{x}_{m+1,j};z\right)|^{s}\middle | \mathcal{F}_{[m',m+1]}^c \right)\\ = \ \gamma^s
\mathbb{E}\left(|G\left(\bm{x}_{m',j'},\bm{x}_{m,\ell};z\right)|^{s}|G^{+}\left(\bm{x}_{m+1,0},\bm{x}_{m+1,j};z\right)|^{s}\middle | \mathcal{F}_{[m',m+1]}^c \right).
\end{multline*}
Integrating first with respect to $\omega(0),...,\omega(m)$, using the inductive assumption and the fact that
$G^{+}\left(\bm{x}_{m+1,0},\bm{x}_{m+1,j};z\right)$ depends only on $\omega(k)$ for $k>m$, we obtain
\begin{multline*}
\mathbb{E}\left(|G\left(\bm{x}_{m',j'},\bm{x}_{m+1,j};z\right)|^{s}\middle | \mathcal{F}_{[m',m+1]}^c  \right) \\ \leq \
\left(\frac{\left(C_{\mathrm{AP}}\gamma^s\right)}{1-s}\right)^{m+1-m'}\mathbb{E}\left(G^{+}\left(\bm{x}_{m+1,0},\bm{x}_{m+1,j};z\right)|^{s}\middle | \mathcal{F}_{[m',m+1]}^c \right).
\end{multline*}
Finally, another application of Lemma \ref{lem:apriori} concludes the proof.
\end{proof}

We now turn to the proof of Lemma \ref{Lemma:decay0n}.  Fix $0< r < \frac{1}{3}$, $z\in \mathbb{C}\setminus \mathbb{R}$, and let $C_{\mathrm{And}}=C_{\mathrm{And}}(3r)$, $\mu_{\mathrm{And}}=\mu_{\mathrm{And}}(3r)$ be such that the Green's function decay for the Anderson model \eqref{eq:fmdecay} holds with $s=3r$.
Let $\mathcal{D}$ denote the boundary layer of $\mathbb{G}_{\mathrm{Diag}}$:  $$\mathcal{D} \ := \ \{(p,p(\ell+1)+j) \ : p\in \mathbb{Z}_{\ge 0} \text{ and } j \in \{0,\ldots,\ell\} \} \ .$$ 
Let us begin by considering the case that neither $\bm{m}$ nor $\bm{n}$ are in $\mathcal{D}$.  It follows that $\bm{m} = \bm{m}'+(j,0)$ and $\bm{n}=\bm{n}'+(k,0)$  with $\bm{m'},\bm{n}'\in \mathcal{D}$ and $j,k\ge1$.  By two applications of the geometric resolvent identity, we find that
\begin{multline*}
G_{\mathrm{Diag}}\left(\bm{m},\bm{n}\right) \ = \ I[\bm{m}'=\bm{n}']  G^{(m_1'+1)}_{\mathrm{And}}(m_1'+j,m_1'+k) \\ + \ G_{\mathrm{And}}^{(m_1'+1)}(m_1'+j,m_1'+1) 
G_{\mathrm{Diag}}\left(\bm{m}',\bm{n}'\right)
G^{(n_1'+1)}_{\mathrm{And}}(n_1'+1,n_1'+k)\ ,
\end{multline*}
where $G_{\mathrm{And}}^{(m)}(j,k) = \langle  \delta_{j}, \, (\mathcal{h}^{(m)}_{\mathrm{And}}- z)^{-1} \delta_{k}\rangle$  is the Green's function of the Anderson model $\mathcal{h}^{(m)}$ and we have surpressed the energy arguments from the Green's functions to simplify notation. 
It follows that
\begin{multline*}\mathbb{E} (| G_{\mathrm{Diag}}(\bm{m},\bm{n}) |^r)  \ \le \ I[\bm{m}'=\bm{n}']\left [ \mathbb{E}(  |G^{(m_1'+1)}_{\mathrm{And}}(m_1'+j,m_1'+k)|^{3r}) \right ]^{\frac{1}{3}} \\
\ + \ \left [ \mathbb{E}  (|G_{\mathrm{And}}^{(m_1'+1)}(m_1'+j,m_1'+1) |^{3r} ) \mathbb{E}(|G_{\mathrm{Diag}}(\bm{m}',\bm{n}') |^{3r})
\mathbb{E} (|G^{(n_1'+1)}_{\mathrm{And}}(n_1'+1,n_1'+k)|^{3r}) \right ]^{\frac{1}{3}} \, 
\end{multline*}
where we have used H\"older's inequality.  Using fractional moment decay \eqref{eq:fmdecay} and Lemma \ref{initialgdecay}, it follows that
\begin{multline}\mathbb{E} (| G_{\mathrm{Diag}}(\bm{m},\bm{n}) |^r)  \ \le \ 
I[\bm{m}'=\bm{n}'] C_{\mathrm{And}}^{\frac{1}{3}} e^{-\frac{1}{3} \mu_{\mathrm{And}}|j-k|} \\ + \ C_{\mathrm{And}}^{\frac{2}{3}} e^{-\frac{1}{3}\mu_{\mathrm{And}}(|j-1|+|k-1|)}  \frac{C_{\mathrm{AP}}^r}{(1-3r)^{\frac{1}{3}}} \left (\frac{(C_{\mathrm{AP}} \gamma)^{r}}{(1-3r)^{\frac{1}{3}}} \right )^{|m_1-n_1|}  \ , \label{eq:doubleri}\end{multline}
where  $\mu_{\mathrm{And}}=\mu_{\mathrm{And}}(3r)$, $C_{\mathrm{And}}=C_{\mathrm{And}}(3r)$.
Let $\alpha=-\log\left(\frac{(C_{\mathrm{AP}} \gamma)^{r}}{(1-3r)^{\frac{1}{3}}}\right)$, $\mu_{\mathrm{Diag}}=\min\bigl(\frac{1}{3}\mu_{\mathrm{And}},\frac{\alpha}{l+2}\bigr)$ and $C_{\mathrm{Diag}}=2\max\bigl ( C_{\mathrm{And}}^{\frac{1}{3}}, C_{\mathrm{And}}^{\frac{2}{3}}\frac{C_{\mathrm{AP}}^r}{(1-3r)^{\frac{1}{3}}}e^{2\mu_{\mathrm{Diag}}(2+\ell)}\bigr )$. In terms of the graph distance in $\mathbb{G}_{\mathrm{Diag},\ell}$, \eqref{eq:doubleri} implies
\begin{equation}\mathbb{E} (| G_{\mathrm{Diag}}(\bm{m},\bm{n}) |^r) \leq 
 C_{\mathrm{Diag}} e^{- \mu_{\mathrm{Diag}}d_{\mathrm{Diag}}(m,n)} \ , \end{equation} so \eqref{decay0n} holds.  Here we made use of the inequality
 \begin{equation}\label{graphineq}
d_{\mathrm{Diag}}(m,n)\leq |j|+|k|+|m'-n'|+\ell
      \end{equation}
      which follows from the triangle inequality and the definition of $\mathbb{G}_{\mathrm{Diag},\ell}.$

If $\bm{m}\in \mathcal{D}$ but $\bm{n}=\bm{n}'+(k,0)$ is not, then similarly we have
\begin{multline*}
    \mathbb{E} (| G_{\mathrm{Diag}}(\bm{m},\bm{n}) |^r) \ \le \ \left[\mathbb{E}(|G_{\mathrm{Diag}}(\bm{m},\bm{n}') |^{3r})
\mathbb{E} (|G^{(n_1'+1)}_{\mathrm{And}}(n_1'+1,n_1'+k)|^{3r}) \right ]^{\frac{1}{3}} \\
\le \ C_{\mathrm{And}}^{\frac{1}{3}} e^{-\frac{1}{3}\mu_{\mathrm{And}}|k-1|} \frac{C_{\mathrm{AP}}^r}{(1-3r)^{\frac{1}{3}}} \left (\frac{(C_{\mathrm{AP}} \gamma)^{r}}{(1-3r)^{\frac{1}{3}}} \right )^{|m_1-n_1|} \ .
\end{multline*}
thus \eqref{decay0n} also holds in this case, as well as when
 $\bm{m}=\bm{m}'+(j,0)$ is not in $\mathcal{D}$ but $\bm{n}$ is, by symmetry.
 Finally, If $\bm{m},\bm{n}\in \mathcal{D}$, then \eqref{decay0n} follows from Lemma \ref{initialgdecay}.
 
 \subsection{Dynamical localization for $H_{\mathrm{Diag}}$}\label{proofthmdiagsec}
 To conclude the proof of Theorem \ref{thmdiag} we will need the following bound.
 \begin{prop}\label{totalvarbound} Let  $\mu_{\bm{n},\bm{m}}$ be the spectral measure of $H_{\mathrm{Diag}}$ associated to $\bm{n}=(n_1,n_2)$ and $\bm{m}=(m_1,m_2)$ and denote its total variation norm by $\|\mu_{\bm{n},\bm{m}}\|$. Let $\mathcal{E}_{m_1}=\{({m_1},m_2)\in \mathbb{G}_{\mathrm{Diag}}\,\,:m_2\in \mathbb{Z}_{+}\}$. Then, for $r\in (0,1)$ and $J=[-2-2\gamma,2+\omega_{\mathrm{max}}+2\gamma]$ we have that
 \begin{equation}\label{totalvardecay}
 \mathbb{E}\left ( \|\mu_{\bm{n},\bm{m}}\| \right ) \ \le \ C \Biggl [  \int_{J} \sum_{\bm{k}\in \mathcal{E}_{m_1}} \mathbb{E} \bigl ( |G_{\mathrm{Diag}}(\bm{n},\bm{k};E)|^{\frac{2r}{1+r}} \bigr )dE \Biggr  ]^{\frac{1+r}{2}}, \
 \end{equation}
  with $C=\norm{\rho}_\infty \omega_{\mathrm{max}}^{\frac{2r}{1+r}}$.
 \end{prop}
 \begin{proof}
 Let $H^L$ be the restriction of $H_{\mathrm{Diag}}$ to $\{\bm{n}\in \mathbb{G}_{\mathrm{Diag}} \ : \ n_1,n_2\le L\}$. Let $\mu_{\bm{n},\bm{m}}^L$ be the spectral measure of $H^L$ associated to $\bm{n}=(n_1,n_2)$ and $\bm{m}=(m_1,m_2)$.

Fix $m_1$ and let $\widehat{H}^L = H^L + (\widehat{v}-\omega({m_1}))P_{m_1}$, where $P_{m_1}$ is the the projection onto $\ell^2(\mathcal{E}_{m_1})$, which has rank 
 $$\bar{d} \ := \ (m_1+1)(\ell+1) \ .$$ Thus $\widehat{H}^L$ is a a copy of $H^L$ with the random potential set to $\widehat{v}\in \mathbb{R}$ at each $\bm{m} \in \mathcal{E}_{m_1}$.   Then, for $\hat v\neq \omega({m_1})$, we have that $H^L\psi_E=E\psi_E$ with $P_{m_1}\psi_E\neq 0$ if and only if $E\notin \sigma\left(\widehat{H}^L\right)$ and
\begin{equation}\label{eq:BS1}\psi_E  \ = \ (\widehat{v}-\omega({m_1}))(\widehat{H}^L- E)^{-1} P_{m_1}\psi_E \ ,
\end{equation}
which implies that
\begin{equation}\label{eq:BS2}P_{m_1}\psi_E  \ = \ (\widehat{v}-\omega({m_1})) P_{m_1} (\widehat{H}^L- E)^{-1} P_{m_1}\psi_E \ .
\end{equation}
Let $\lambda_1(E)$, $\cdots$, $\lambda_{\bar{d}}(E)$ denote the eigenvalues of $P_{m_1} (\widehat{H}^L- E)^{-1} P_{m_1}  $, with corresponding normalized eigenvectors $\phi_1(E,\cdot),\ldots,\phi_{\bar{d}}(E,\cdot)$ (which we regard as functions of $E\in \mathbb{R}$ and $k\in \mathbb{G}_{\mathrm{Diag}}$). Choose branches so that these functions are continuously differentiable in $E$ away from the finite set of poles (the eigenvalues of $\widehat{H}^L$). To simplify notation, we sometimes write $\phi_j(E)$ to denote the function $k\mapsto \phi_j(E,k)$. 

From \eqref{eq:BS2}, we conclude that if $\psi_E$ is a normalized eigenvector of $H^L$ with eigenvalue $E$, then $P_{m_1}\psi_E=c\phi_j(E)$ with $c\neq 0$ and $(\widehat{v}-\omega({m_1}))\lambda_j(E)=1$ for some $j\in\{1,\cdots,\bar{d}\}$. By \eqref{eq:BS1} this implies $$\psi_E= c (\widehat{v}-\omega({m_1})) (\widehat{H}^L-E)^{-1}P_{m_1} \phi_j(E) \ .$$ The normalization constant $c$ is given by
$$\frac{1}{(\widehat{v}-\omega({m_1}))^2c^2} = \left \| (\widehat{H}^L - E)^{-1} P_{m_1} \phi_j(E)\right \|^2 
= \langle \phi_j(E), P_{m_1} (\widehat{H}^L - E)^{-2} P_{m_1}\phi_j(E)\rangle \ = \ \lambda_j'(E)  \, $$
where in the last step we made use of the Hellmann-Feynman Theorem.

It follows from the above considerations that the spectral measure $\mu^L_{\bm{n},\bm{{m}}}$ for $H^L$ associated to $\bm{n}$ and $\bm{{m}}$ is given by 
\begin{align}\nonumber
  d\mu^L_{\bm{n},\bm{m}}(E)
  \ &= \ \sum_{j=1}^{\bar{d}}\psi_E(\bm{n})\psi_E(\bm{m}) \, \delta\left ((\widehat{v}-\omega({m_1}))\lambda_j(E) -1 \right ) \, | \widehat{v}-\omega({m_1}) | \lambda_j'(E)\,dE \\
  &= \ 
  \sum_{j=1}^{\bar{d}}   g_j(E,\bm{n})g_j(E,\bm{m}) \,  \delta\left ((\widehat{v}-\omega({m_1}))\lambda_j(E) -1 \right) \,  | \widehat{v}-\omega({m_1}) | dE  \ ,  \label{eq:munm}
\end{align}
where $g_j(E,\bm{n}) =  \langle \delta_{\bm{n}}, (\widehat{H}-E)^{-1}P_{m_1} \phi_j(E) \rangle $. Observe that
\begin{equation}\label{eq:gtophi} g_j(E,\bm{m})=\lambda_j(E)\phi_j(E,\bm{m})=\frac{1}{(\widehat{v}-\omega({m_1}))}\phi_j(E,\bm{m}) \ , \end{equation} since $P_{m_1}\delta_{\bm{m}}=\delta_{\bm{m}}$. 
We note that 
$$ \sum_{j=1}^{\bar{d}}   \int_J (g_j(E,\bm{m}))^2 \,  \delta\left ((\widehat{v}-\omega({m_1}))\lambda_j(E) -1 \right) \,  | \widehat{v}-\omega({m_1}) | dE \ = \ \int_{J} d\mu_{\bm{m},\bm{m}}^L(E) \ = \ 1 , $$
since $\sigma\left(H_{\mathrm{Diag}}\right)\subset J =[-2-2\gamma,2+\omega_{\mathrm{max}}+2\gamma]$ for all realizations of the random potential.
Furthermore, we have
\begin{multline}\label{eq:l2normg} \sum_{j=1}^{\bar{d}} \int_{J} \left ( g_j(E,\bm{n}) \right )^2   \, \delta((\widehat{v}-\omega({m_1}))\lambda_j(E)  -1) \, |\widehat{v}-\omega(m_1)| dE \\ = \
\int_{\mathbb{R}} \left ( \frac{d \mu_{\bm{n},\bm{m}}^L}{d \mu_{\bm{m},\bm{m}}^L}(E) \right )^2 d \mu_{\bm{m},\bm{m}}^L(E)\ \le \ 1 \ , \end{multline}
since the integral on the right-hand side gives the norm-squared of the projection of $\delta_{\bm{n}}$ onto the cyclic subspace for $H^L$ generated by $\delta_{\bm{m}}$.

It follows from \eqref{eq:munm} and the fact that $\sigma(H^L)\subset J$ that the total variation norm of $\mu_{\bm{n},\bm{m}}^L$ is given by
\begin{equation}\label{eq:totalvar} \| \mu_{\bm{n},\bm{m}}^L\| \ = \ \sum_{j=1}^{\bar{d}} \int_{J} |g_j(E,\bm{n})| |g_j(E,\bm{m})| 
\delta((\widehat{v}-\omega({m_1}))\lambda_j(E)  +1) \,  |\widehat{v}-\omega(m_1)| dE \ . \end{equation}
Let $\widehat{\mathbb{E}}$ denote averaging over $\omega$ and an arbitrarily chosen distribution for $\widehat{v}$ (which we will take below to be the same as the distribution for $\omega(m_1)$). It follows from \eqref{eq:totalvar} and \eqref{eq:l2normg} that 
\begin{multline} \label{eq:firsttvbound} \widehat{\mathbb{E}}( \|\mu_{\bm{n},\bm{m}}^L\|) \ \le \ \Biggl [ 
\widehat{\mathbb{E}} \Biggl ( \sum_{j=1}^{\bar{d}} \int_{J}
|g_j(E,\bm{n})|^{\frac{2r}{1+r}} |g_j(E,\bm{m})|^{\frac{2}{1+r}}  \\ \times \delta((\widehat{v}-\omega({m_1}))\lambda_j(E) -1) \, |\widehat{v}-\omega(m_1)| dE \Biggr ) \Biggr ]^{\frac{1+r}{2}} \ , \end{multline}
where $0<r<1$ and we have used H\"older's inequality with exponents $p=\frac{2}{1+r}$ and $q=\frac{2}{1-r}$ (applied first to the integral over $J$ and then to the expectation).

We wish to rewrite \eqref{eq:firsttvbound} in terms of the normalized eigenfunctions $\phi_j(E)$ of $P_{m_1}(\widehat{H}^L -E)^{-1} P_{m_1}$.  Noting that
$g_j(E,\bm{n})  =  \sum_{\bm{k}\in \mathcal{E}_{m_1}} \widehat{G}(\bm{n},\bm{k};E) \phi_j(E,\bm{k}) $
and using \eqref{eq:gtophi}, we find that
\begin{multline*}
    \widehat{\mathbb{E}}\left ( \| \mu_{\bm{n},\bm{m}}^L\| \right ) \ \le \  \Biggl [ \widehat{\mathbb{E}} \Biggl (  \sum_{j=1}^{\bar{d}}  \int_{J} \sum_{\bm{k}\in \mathcal{E}_{m_1}}  |\widehat{G}(\bm{n},\bm{k};E)|^{\frac{2r}{1+r}} \frac{1}{|\widehat{v}-\omega(m_1)|^{\frac{2}{1+r}}} |\phi_j(E,\bm{k})|^{\frac{2r}{1+r}} |\phi_j(E,\bm{m})|^{\frac{2}{1+r}}   \\  
    \times  \delta((\widehat{v}-\omega({m_1}))\lambda_j(E) -1) \, |\widehat{v}-\omega(m_1)| dE \Biggr ) \Biggr  ]^{\frac{1+r}{2}} \ .
\end{multline*}
Now integrate first with respect to $\omega({m_1})$, assuming that $0\leq \hat v\leq \omega_{\max}$, using
\begin{multline}
   \int_{0}^{\omega_{\mathrm{max}}} |\widehat{v}-\omega({m_1})|^{1-\frac{2}{1+r}}\delta((\widehat{v}-\omega({m_1}))\lambda_j(E) -1)\, \rho(\omega(m_1))  d\omega({m_1})\\
   = \ \int_{0}^{\omega_{\mathrm{max}}} |\widehat{v}-\omega({m_1})|^{\frac{2r}{1+r}}\delta\bigl (\widehat{v}-\omega({m_1}) -\frac{1}{\lambda_j(E)} \bigr )\, \rho(\omega({m_1})) d\omega({m_1}) \ \leq \ \norm{\rho}_\infty \,  \omega_{\mathrm{max}}^{\frac{2r}{1+r}} \ ,
   \end{multline}
   where we have used that $\delta(ax)=\frac{1}{|a|}\delta(x)$ and that $\lambda_j(E)(\widehat{v}-\omega(m_1))=1$.
Thus, letting $C=\norm{\rho}_\infty \omega_{\mathrm{max}}^{\frac{2r}{1+r}}$, we have
\begin{equation} \widehat{\mathbb{E}}\left ( \| \mu_{\bm{n},\bm{m}}^L\| \right ) \ \le C \Biggl [  \int_{J} \sum_{\bm{k}\in \mathcal{E}_{m_1}} \widehat{\mathbb{E}} \Biggl  ( |\widehat{G}(\bm{n},\bm{k};E)|^{\frac{2r}{1+r}} \sum_{j=1}^{\bar{d}} |\phi_j(E,\bm{k})|^{\frac{2r}{1+r}} |\phi_j(E,\bm{m})|^{\frac{2}{1+r}} \Biggr )\Biggr ]^{\frac{1+r}{2}} \ . \label{eq:secondtv}
\end{equation}

Since $\{\phi_j(E)\}_{j=1}^{\bar{d}}$ is an orthonormal basis for $\ell^2(\mathcal{E}_{m_1})$, we have, by H\"older's inequality,
\begin{align*} \sum_{j=1}^{\bar{d}} |\phi_j(E,\bm{k})|^{\frac{2r}{1+r}} |\phi_j(E,\bm{m})|^{\frac{2}{1+r}} \ &\le \ \left ( \sum_{j=1}^{\bar{d}} |\phi_j(E,\bm{k})|^{2} \right )^{\frac{r}{1+r}} \left ( \sum_{j=1}^{\bar{d}}  |\phi_j(E,\bm{m})|^{2} \right )^{\frac{1}{1+r}}\\
&\ = \ 1 \ .\end{align*}
Therefore, it follows from \eqref{eq:secondtv} that
$$  \widehat{\mathbb{E}}\left ( \| \mu_{\bm{n},\bm{m}}^L\| \right ) \ \le \ C \Biggl [  \int_{J} \sum_{\bm{k}\in \mathcal{E}_{m_1}} \widehat{\mathbb{E}} \bigl ( |\widehat{G}(\bm{n},\bm{k};E)|^{\frac{2r}{1+r}} \bigr )\Biggr ]^{\frac{1+r}{2}} \ . $$
Choosing the distribution of $\widehat{v}$ to be identical to that of $\omega({m_1})$, and independent from $\omega$, we find that
\begin{equation} \widehat{\mathbb{E}}\left ( \| \mu_{\bm{n},\bm{m}}^L\| \right ) \ \le \ C \Biggl [  \int_{J} \sum_{\bm{k}\in \mathcal{E}_{m_1}} \mathbb{E} \bigl ( |G(\bm{n},\bm{k};E)|^{\frac{2r}{1+r}} \bigr )\Biggr ]^{\frac{1+r}{2}} \, , \label{finvolvar} 
\end{equation}
where $G$ denotes the Green's function of $H^L$.
Since $ \mu_{\bm{n},\bm{m}}^L$ converges in the vague topology to $\mu_{\bm{n},\bm{m}}$, \eqref{totalvardecay} follows from \eqref{finvolvar} and Fatou's Lemma. \qedhere
\end{proof}
Theorem \ref{thmdiag} immediately follows from Proposition \ref{totalvarbound} and Lemma \ref{Lemma:decay0n} since $\mu_{\bm{n},\bm{m}}$ is a regular Borel measure and thus for any Borel set $F\subset \mathbb{R}$ we have that $$|\mu_{\bm{n},\bm{m}}|(F)=\sup_{|f|\leq 1}\Big|\int_{F} f(x)\,d\mu_{\bm{n},\bm{m}}(x)\Big|,$$ with the supremum taken over Borel measurable functions $f$ bounded by one.

\appendix
\section{A Version of Boole's Equality for level sets of Herglotz functions: Proof of proposition \ref{booletype}}\label{sec:Booleproof}
In this appendix we prove Proposition \ref{booletype}, recalled here for the reader's convenience:  
\begin{quote}
    \textbf{Proposition 7} \it Let $\mu$ be a finite Borel measure which is purely singular and let $F(z)=\int \frac{1}{u-z}\,d \mu(u) $ be its Borel transform. Then
 \begin{equation}\tag{\ref{Boole's}}
 \big|\{E\in \mathbb{R}\,\,:\alpha< E+F(E+i0)<\beta\}\big|=\beta-\alpha.
 \end{equation}
\end{quote}

It is instructive to consider first the situation when $\mu$ is a pure point measure with finitely many atoms, in which case $F$ is a rational function of the form $F(E)=\sum^{N}_{n=1}\frac{p_n}{u_n-E}$ with real poles $\{u_n\}_{n=1}^{N}$ at the atoms of $\mu$ and weights $\{p_n\}_{n=1}^{N}$ with $p_n=\mu({u_n})$. For example, the diagonal elements of the Green's function in finite volume are of this form. For a real number $\lambda$, let $Q_{\lambda}$ be a polynomial of degree $N+1$
given by
\begin{equation*}Q_{\lambda}(E)=(\lambda-E-F(E))\prod^{N}_{n=1}(E-u_n).\end{equation*}
The solutions $v_1(\lambda),...,v_{N+1}(\lambda)$ of the equation $E+F(E)=\lambda$ coincide with the roots of $Q_{\lambda}$. Therefore, the coefficient of
 $E^N$  in 
 $Q_{\lambda}(E)=-\prod^{N+1}_{n=1}(E-v_n(\lambda))$ equals $\sum^{N+1}_{n=1}v_n(\lambda)$. On the other hand, by definition of $Q_{\lambda}$, this coefficient
 is $\lambda+\sum^{N+1}_{n=1}u_n$.  Therefore, 
\begin{equation}\label{firsteqroots}\sum^{N+1}_{n=1}v_n(\lambda)=\lambda +\sum^{N}_{n=1}u_n.
\end{equation}
Since $E+F(E)$ is monotone increasing between poles, the set $\{E\in \mathbb{R} : \alpha < E+F(E) < \beta\}$ is a disjoint union of intervals $\cup^{N+1}_{n=1}(v_n(\alpha),v_n(\beta))$. Therefore, we
conclude from equation (\ref{firsteqroots}) that
$$
\big|\{E\in \mathbb{R}\ : \ \alpha \ < \ E+F(E)\ <\  \beta\}\big| \ =\ \sum^{N+1}_{n=1} (v_n(\beta)-v_n(\alpha))
\ = \ \beta-\alpha \ .$$

The above proof is not readily generalized to other types of measures.  The following argument is inspired by the analysis in \cite[Proposition 8.2]{A-W-B} and provides a proof which is valid for general singular measures.

\begin{proof}[Proof of Proposition \ref{booletype}]
The function $F(z)$ is a Herglotz function, i.e., a holomorphic map from the upper half plane to itself.  It follows from the classical theory of such functions (see \cite[Theorem 5.9.1]{Simon-Book}) that for almost every $E \in \mathbb{R}$,
\begin{enumerate}
    \item the boundary value $F(E+i0)=\lim_{\epsilon \rightarrow 0} F(E+i\epsilon)$ exists, and
    \item $F(E+i0)$ is real (because $\mu$ is singular).
\end{enumerate}
Furthermore, for any $\lambda\in \mathbb{R}$ the level set $\{E \ : \ E+F(E+i0) = \lambda \}$ is a Lebesgue null set.  To see this note that $G(z)=(\lambda - z - F(z))^{-1}$ is a Herglotz function, so its boundary value $G(E+i0)$ exists for almost every $E\in \mathbb{R}$. Thus $F(E+i0)+E\neq \lambda$ for almost every $E$.

It follows from the above considerations that the indicator function of the set $\{E:\alpha<E+F(E+i0)<\beta\}$ can be represented, for almost every $E$, as
\begin{equation}\label{eq:1=imlog} \mathds{1}[E \ :\ \alpha\ <\ E+F(E+i0)\ <\ \beta ] \ = \ \lim_{\epsilon \downarrow 0} \frac{1}{\pi}u_{\alpha,\beta}(E+i\epsilon) \ ,\end{equation} where
$u_{\alpha,\beta}(z) = \mathrm{Im}\log (z + F(z)-\beta)-\mathrm{Im}\log(z+F(z)-\alpha)$ and $\log$ denotes the principal branch of the logarithm.  The function $u_{\alpha,\beta}(z)$ is harmonic and bounded by $\pi$ for $z$ in the upper half plane.  By dominated convergence and the fact that a harmonic function is reproduced by the Poisson integral of its boundary values over a half plane (see \cite[Theorem 5.9.2]{Simon-Book}), we have
\begin{align*}|\{E\in \mathbb{R} \ : \ \alpha < E+F(E+i0)<\beta\}| \ &= \ \lim_{\epsilon \to 0} \frac{1}{\pi}\int_{\mathbb{R}} u_{\alpha,\beta}(E+i\epsilon) d E \\
&= \ \lim_{\epsilon \to 0} \lim_{\eta\rightarrow \infty} \frac{1}{\pi}\int_{\mathbb{R}} \frac{\eta^2}{E^2+\eta^2} u_{\alpha,\beta}(E+i\epsilon) dE \\ 
&= \ \lim_{\epsilon \to 0} \lim_{\eta \rightarrow \infty} \eta u_{\alpha,\beta}(i(\eta + \epsilon)) \\ 
&= \ \lim_{\eta \rightarrow \infty} \eta u_{\alpha,\beta}(i\eta) \ ,
\end{align*}
where the equality follows since $\limsup_{\epsilon \rightarrow 0} \limsup_{\eta \rightarrow \infty} \epsilon |u_{\alpha,\beta}(i(\eta+\epsilon))|=0$, because $u_{\alpha,\beta}$ is bounded. On the other hand, by definition of $u_{\alpha,\beta}$ we know that
\begin{equation}u_{\alpha,\beta}(i\eta) \ = \ \mathrm{Im}\int^{\beta}_{\alpha}\frac{1}{E-i\eta -F(i\eta)}\, dE
\end{equation}
Hence, by dominated convergence again,
$$
\lim_{\eta \to \infty}\eta u_{\alpha,\beta}(i\eta) \
= \ \lim_{\eta \to \infty}\int^{\beta}_{\alpha}\frac{\eta^2+\eta\mathrm{Im}F(i\eta)}{(E-\mathrm{Re}F(i\eta))^2+(\eta
+\mathrm{Im}F(i\eta))^2}\,dE \ = \ \beta-\alpha
$$
where we have used the simple facts that $\lim_{\eta\to \infty} F(i\eta)=0$ and $\lim_{\eta\to
\infty}\eta\mathrm{Im}F(i\eta)=\mu(\mathbb{R}).$
\end{proof}

\section{Localization in the horizontal direction}\label{sec:appendixB}
In this section we sketch the proof of Lemma \ref{horloc}, which recall here for the reader's convenience:
\begin{taglem}{\ref{horloc}}
Given $s\in (0,1)$ there exist positive constants $C_{\mathrm{And}}(s)$ and $\mu_{\mathrm{And}}=\mu_{\mathrm{And}}(s)$ such that, for all $m,n \in \mathbb{Z}_{\geq 0}$,
\begin{equation}\tag{\ref{dynlocand}}
\mathbb{E}\left(\sup_{|f|\leq 1}|\langle \delta_m,f\left(\mathcal{h}_{p}\right)\delta_n\rangle|\right)\ \leq \ 
(2\omega_{\max}+2\gamma)^{\frac{s}{2-s}}
\left(\|\rho\|_{\infty}\left(4+4\gamma+\omega_{\max}\right)C_{\mathrm{And}}(s)\right)^{\frac{1}{2-s}}\, e^{-\frac{\mu_{\mathrm{And}}}{2-s}|m-n|} \ ,
\end{equation}
with the supremum taken over all Borel measurable functions bounded by one.
\end{taglem}

\begin{proof}
 Recall that, according to \eqref{defhp} $\mathcal{h}_{p}= \mathcal{h}_{\mathrm{And}}^{(0)} - 2\gamma \cos p P_{0}$ with $P_0$ the projection onto $\delta_0$.  We follow closely the proof of \cite[Theorem A1]{A-S-F-H}. It suffices to show that, for every $L\in \mathbb{N}$, (\ref{dynlocand}) holds with $\mathcal{h}_{p}$ replaced by its restriction to
$\ell^{2}\left(\mathbb{Z}_{+}\cap[0,L]\right)$, denoted henceforth by $\mathcal{h}^{L}_{p}$. Note that $\mathcal{h}^L_p\rightarrow \mathcal{h}_p$ in the strong resolvent sense as $L\rightarrow \infty$.
Let $\hat\omega(0)\in \mathbb{R}$ be arbitrary and $$\hat{\mathcal{h}}_{\mathrm{And}}^{(0)}=\mathcal{h}_{\mathrm{And}}^{(0)}+
\left(\hat \omega(0)-\omega(0)+2 \gamma \cos p\right)P_0$$ be a copy of $\mathcal{h}_{\mathrm{And}}^{(0)}$ with the random potential at zero set to $\hat\omega(0)$. Write $v_p=\omega(0)-2 \gamma \cos p$. From rank-one perturbation formulas (see, for instance, \cite[Theorem 5.3]{A-W-B} or
\cite[Equation(A.7)]{A-S-F-H}), the spectral measure of $\mathcal{h}^{L}_{p}$ is given by
\begin{equation}\label{rankone}
d\mu^{p,L}_{\bm{m},\bm{n}}(E)=\left(\hat\omega(0)-v_p\right) {\hat G}^{L}_{\mathrm{And}}\left(m,n;E\right)\delta(v_p-\hat\omega(0)-\Sigma^L(E))\,dE
\end{equation}
where
${\hat G}^{L}_{\mathrm{And}}\left(m,n;E\right)$ is the the Green's function of the Anderson model $\hat{\mathcal{h}}_{\mathrm{And}}^{(0)}$ restricted to $\ell^{2}\left(\mathbb{Z}_{+}\cap[0,L]\right)$ and
\begin{equation}
\Sigma^L(E):=-\frac{1}{{\hat G}^{L}_{\mathrm{And}}\left(0,0;E\right)}.
\end{equation}
Equation \eqref{rankone} implies a couple of estimates. The first one is obtained letting $m=n$ in (\ref{rankone}) to achieve
\begin{equation}
d\mu^{p,L}_{m,m}(E)=\delta(v_p-\hat\omega(0)-\Sigma^L(E))\,dE.
\end{equation}
In particular
\begin{equation}\label{probweight}
\int^{\infty}_{-\infty}\delta(v_p-\hat\omega(0)-\Sigma^L(E))\,dE=1.
\end{equation} A second observation is that 
\begin{equation}\label{elltwo}
|\hat\omega(0)-v_p|^2  \int \left|{\hat G}^{L}_{\mathrm{And}}\left( m,n;E\right)\right|^2\delta(v_p-\hat\omega(0)-\Sigma^L(E))\,dE\leq 1.
\end{equation}

Indeed, as explained in \cite[Equation(A.9)]{A-S-F-H}, $d\mu^{p,L}_{m,n}(E)=\psi(E)d\mu^{p,L}_{m,m}(E)$ with $$\int |\psi(E)|^2\,d\mu^{p,L}_{m,m}(E)=\langle \delta_n,P_{\delta_m}\delta_n\rangle\leq 1,$$ where $P_{\delta_m}$ is the projection onto the cyclic subspace of $\mathcal{h}^{L}_{p}$ which contains $\delta_m$. 

Combining equations (\ref{rankone}), (\ref{elltwo}) and (\ref{probweight}) with H\"older's inequality (applied to the exponents $(p,q)=(2-s,\frac{2-s}{1-s})$) and Jensen's inequality for expectations, we conclude
that for all intervals $I \subset \mathbb{R}$
\begin{equation}\mathbb{E}\left(\left|\mu^{p,L}_{m,n}\right|(I) \right)\leq \left[
\mathbb{E}\left(|\hat\omega(0)-v_p|^s\int_{I}\left|{\hat G}^{L}_{\mathrm{And}}\left(m,n;E \right)\right|^s\delta(v_p-\hat\omega(0)-\Sigma^L(E))\,dE \right)\right]^{\frac{1}{2-s}}.
\end{equation}
Thus,
$$\mathbb{E}\left(\left|\mu^{p,L}_{m,n}\right|(I) \right)\leq
\left(2\omega_{\max}+2\gamma\right)^{\frac{s}{2-s}}\left(\int_{I}
\mathbb{E}\left(\left|{\hat G}^{L}_{\mathrm{And}}\left(m,n;E\right)\right|^s\right)\delta(v_p-\hat\omega(0)-\Sigma^L(E))\,dE \right)^{\frac{1}{2-s}}.$$ 
Recalling that $v_p=\omega(0)-2 \gamma \cos p$, integrating first over $\omega(0)$ and choosing $\hat\omega(0)$ to be a random variable independent of $\omega(0)$ but identically distributed with it we conclude that

\begin{equation}\label{hordecayfin}
\mathbb{E}\left(\left|\mu^{p,L}_{m,n}\right|(I) \right)\leq
\left(2\omega_{\max}+2\gamma\right)^{\frac{s}{2-s}}\|\rho\|^{\frac{1}{2-s}}_{\infty}\left(\int_{I}
\mathbb{E}\left(\left|G^{L}_{\mathrm{And}}\left(m,n;E\right)\right|^s\right)\,dE \right)^{\frac{1}{2-s}}.
\end{equation}
Since the operator $\mathcal{h}_{p}$ has spectrum contained in $[-2-2\gamma, 2+2\gamma+\omega_{\mathrm{max}}]$, the inequality
(\ref{hordecayfin}) together with \eqref{eq:fmdecay} suffices to conclude the proof of lemma \ref{horloc}.
We mention that by introducing an integrable weight, one could also handle the case where the random potentials are unbounded. For further details we refer to
\cite[Equations (A.13)-(A.18)]{A-S-F-H}.\qedhere
\end{proof}

\section{\emph{A priori} bounds on the Green's function}\label{sec:apriori}
Let $\mathcal{H}$ and $\mathcal{H}_1$ be separable Hilbert spaces and let $A:D(A)\subset \mathcal{H}\rightarrow \mathcal{H}$ be a maximally dissipative operator. Recall that a densely defined operator $A$ is said to be dissipative if $\mathrm{Im}\langle \varphi, A \varphi \rangle \geq 0$ for every $\varphi \in D(A)$. $A$ is said to be maximally dissipative when it is dissipative and has no proper dissipative extension.
Let $M_1:\mathcal{H}\rightarrow \mathcal{H}_1$ and $M_2 :\mathcal{H}_1\rightarrow \mathcal{H}$ be  Hilbert-Schmidt operators.
 Denoting by $|\,.\,|$ Lebesgue measure and by $\|\,\cdot\,\|_{HS}$ the Hilbert-Schmidt norm, the following weak $L_1$ bounds hold

\begin{lem}{\cite[Lemma 3.1]{A-E-N-S-S}}
\begin{equation}\label{weakl1}\Big|\{v \,\,:\|M_1\frac{1}{A-v+i0}M_2\|_{HS}>t\}\Big|\leq C_W \|M_1\|_{HS}\|M_2\|_{HS}\frac{1}{t}
\end{equation}
where the constant $C_{W}$ is independent of $A$,$M_1$ and $M_2$.
\end{lem}

\begin{lem}\cite[Proposition 3.2]{A-E-N-S-S}
Let $A$,$M_1$ and $M_2$ be as above and let $U_1,U_2$ be nonnegative operators.
\begin{equation}\label{weakrank2}\Big|\{(v_1,v_2)\in [0,1]^2 \,\,:\|M_1U^{1/2}_{1}\frac{1}{A-v+i0}U^{1/2}_{2}M_2\|_{HS}>t\}\Big|\leq 2C_W
\|M_1\|_{HS}\|M_2\|_{HS}\frac{1}{t}
\end{equation}
\end{lem}

The bounds (\ref{weakl1}), (\ref{weakrank2}) easily imply the \emph{apriori} bound of Lemma \eqref{eq:apriori}
\begin{equation}\mathbb{E} \left ( \left | G_{\mathrm{Diag}}(\bm{m},\bm{n};z) \right |^s \middle | \mathcal{F}_{m_1,n_1}^c \right ) \ \le \ \frac{C_{\mathrm{AP}}^s}{1-s} \ ,
\end{equation}
where we recall that $\mathbb{E}\left ( \cdot \middle | \mathcal{F}_{m_1,n_1}^c \right )$ denotes averaging with respect to the variables $\omega(m_1)$ and $\omega(n_1)$. For further details refer the reader to \cite[appendix A]{M-S}.

\vspace{0.5cm}

\textbf{Acknowledgments.} R. Matos is thankful to Wencai Liu and Shiwen Zhang for useful discussions. This work was supported by the National Science Foundation under grants no 1900015 and 2000345.

\bibliographystyle{amsplain}

\begin{thebibliography}{10}

\bibitem{A-E-N-S-S} M. Aizenman, A. Elgart, S. Naboko, J. Schenker, and G. Stolz, \newblock
``Moment analysis for localization in random Schr\"odinger operators'', \newblock
\textit{Invent. Math.} \textbf{163} (2006), 343-413.

\bibitem{A-M} M. Aizenman and S. Molchanov,  \newblock
``Localization at large disorder and extreme energies: an elementary derivation'', \newblock
\textit{Commun. Math. Phys.} \textbf{157} (1993), 245-278.
 
\bibitem{A-S-F-H} M. Aizenman, J.H Schenker, R.M. Friedrich, and D. Hundertmark. \newblock 
``Finite-Volume Fractional-Moment Criteria for Anderson Localization'', \newblock 
\textit{Commun. Math. Phys.} \textbf{224} (2001), 219-254.

\bibitem{A-W-Ballistic} M. Aizenman and S. Warzel, \newblock 
``Absolutely continuous spectrum implies ballistic transport for quantum particles in a random potential on tree graphs'', \newblock 
\textit{Journal of mathematical physics} \textbf{53} (2012) 095205 .
     
\bibitem{A-W-B} M. Aizenman and S. Warzel, \newblock 
\textit{Random Operators: Disorder Effects on Quantum Spectra and Dynamics}, \newblock 
Graduate Studies in Math. \textbf{18},  Amer. Math. Soc., Providence, RI, 2015.
 
\bibitem{A-W-P}  M. Aizenman and S. Warzel, \newblock 
``Localization Bounds for Multiparticle Systems'', \newblock 
\textit{Commun. Math. Phys.} \textbf{290} (2009), 903-934.
 
\bibitem{A} P. Anderson, \newblock 
``Absence of Difusion in Certain Random Lattices'', \newblock 
\textit{Phys. Rev.} \textbf{109} (1958), 1492-1505.

\bibitem{A-S} J. Avron and B. Simon, \newblock 
``Transient and Recurrent Spectrum.'', \newblock 
\textit{J. Funct. Anal.} \textbf{43} (1981), 1-31.

\bibitem{Boole}G. Boole, \newblock 
``On the comparison of transcendents, with certain applications to the theory of definite integrals'', \newblock 
\textit{Philos. Trans. Roy. Soc. London} \textbf{147} (1857), 745–803.

\bibitem{C-K-M} R. Carmona, A. Klein and F. Martinelli, \newblock 
``Anderson Localization for Bernoulli and other singular potentials'', \newblock
\textit{Commun. Math. Phys.} \textbf{108} (1987),  41-66.

\bibitem{Dam} D. Damanik, V. Bucaj, J. Fillman, V.  Gerbuz, T. VandenBoom, F. Wang and Z. Zhang, \newblock 
``Localization for the one-dimensional Anderson model via positivity and large deviations for the Lyapunov exponent'', \newblock 
\textit{Trans. Amer. Math. Soc.} \textbf{372} (2019), 3619-3667

\bibitem{damanik2015} D. Damanik, M. Lukic, W. Yessen, \newblock
``Quantum {{Dynamics}} of {{Periodic}} and {{Limit}}-{{Periodic Jacobi}} and {{Block Jacobi Matrices}} with {{Applications}} to {{Some Quantum Many Body Problems}}'', \newblock
\textit{Commun. Math. Phys.} \textbf{337} (2015), {1535-1561}

\bibitem{Davis1} B. Davis, \newblock 
``On the distributions of conjugate functions of nonnegative measures'', \newblock 
\textit{Duke Math. J.} \textbf{40} (1973), 695–700.

\bibitem{Davis2}  B. Davis, \newblock 
``On the weak type (1, 1) inequality for conjugate functions'', \newblock 
\textit{Proc. Amer. Math. Soc.} \textbf{44} (1974), 307–311.

\bibitem{delrio1995} R. del Rio, S. Jitomirskaya, Y. Last and B. Simon, \newblock 
``What Is {{Localization}}?'', \newblock 
\textit{Phys. Rev. Lett.} \textbf{75} (1995), 117-119.

\bibitem{D-J-L-S} R. del Rio, S. Jitomirskaya, Y. Last and B. Simon, \newblock 
``Operators with singular continuous spectrum IV Hausdorff dimension, rank one perturbations and localization'', \newblock 
\textit{J. d’Anal. Math.} \textbf{69} (1996), 153-200.
    
\bibitem{D-S} R. del Rio and B. Simon,
\newblock ``Point Spectrum and Mixed Spectral Types for Rank One Perturbations'', \newblock
\textit{Proc. Amer. Math. Soc.} \textbf{125} (1997), 3593-3599.

\bibitem{F-S} J. Fr\"ohlich and T. Spencer, \newblock 
``Absence of diffusion in the Anderson tight binding model for large disorder or low energy'', \newblock 
\textit{Commun. Math. Phys.}  \textbf{88} (1983), 151-184.

\bibitem{GMP} I. Goldsheid, S. Molchanov, L. Pastur, \newblock 
``A pure point spectrum of the one-dimensional Schr\"odinger operator'', \newblock 
\textit{Funct. Anal. Appl.} \textbf{11} 1977, 1–10.

\bibitem{Guarnieribound} I. Guarnieri, \newblock 
``On an estimate concerning quantum diffusion in the presence of a fractal spectrum'', \newblock 
\textit{Europhys. Lett.} \textbf{21} (1993), 729-733.

\bibitem{G-SB} I. Guarnieri and H. Schulz-Baldes, \newblock 
``Lower Bounds on Wave Packet Propagation By Packing Dimensions of Spectral Measures'', \newblock 
\textit{Math. Phys. Elect. J.} \textbf{5} (2002), 1-16.

\bibitem{Vinogradov} S. V. Hrušcev and S. A. Vinogradov, \newblock 
``Free interpolation in the space of uniformly convergent Taylor series'', \newblock 
pp. 171-213 in \textit{Complex Analysis and Spectral Theory (Leningrad, 1979/1980)}, \newblock 
Lecture Notes in Math. \textbf{864}, Springer, Berlin–New York, 1981.

\bibitem{JMdimtwo} V. Jakši\'c and S. Molchanov, \newblock 
``On the Surface Spectrum in Dimension Two'', \newblock 
\textit{Helv. Phys. Acta} \textbf{71} (1999), 629-657.

\bibitem{JMMM} V. Jakši\'c and S. Molchanov, \newblock 
``On the Spectrum of the Surface Maryland Model'', \newblock 
\textit{Lett. Math. Phys.} \textbf{45} (1998), 185-193.

\bibitem{J-M} V. Jakši\'c and S. Molchanov, \newblock 
``Localization of Surface Spectra'', \newblock 
\textit{Comm. Math. Phys.} \textbf{208} (1999), 153-172.

\bibitem{jitomirskaya2003} S. Jitomirskaya, H. {Schulz-Baldes} and G.  Stolz, \newblock
``Delocalization in Random Polymer Models'', \newblock
\textit{Commun. Math. Phys.} \textbf{233} (2003), 27-48.

\bibitem{J-B} S. Jitomirskaya and H. Schulz-Baldes, \newblock
``Upper Bounds for Wavepacket spreading for random Jacobi matrices'', \newblock 
\textit{Comm. Math. Phys.} \textbf{273} (2007), 601-618.


\bibitem{J-Z}S. Jitomirskaya and X. Zhu, \newblock 
``Large Deviations of the Lyapunov Exponent and Localization for the 1D Anderson Model'', \newblock
\textit{Commun. Math. Phys.} \textbf{370} (2019), 311–324. 

\bibitem{K-L-S-S} Y. Karpeshina, YR Lee., R. Shterenberg, G. Stolz \newblock
``Ballistic Transport for the Schrödinger Operator with Limit-Periodic or Quasi-Periodic Potential in Dimension Two." \newblock
\textit{Commun. Math. Phys.} \textbf{354} (2017), 85–113.



\bibitem{Klein} A. Klein, \newblock
``Multiscale Analysis and Localization of Random Operators'', \newblock
arXiv:0708.2292 (2007).

\bibitem{K-K-L} R. Killip, A. Kiselev Y. Last, \newblock
``Dynamical Upper Bounds for Wavepacket Spreading'', \newblock
\textit{Am. J. Math.} \textbf{125} (2003), 1165-1198.

\bibitem{K-S} H. Kunz and B Souillard, \newblock
``Sur le spectre des operateurs aux differences finies aleatories'' \newblock
\textit{Commun. Math. Phys.} \textbf{78} (1981), 201-246.

\bibitem{L-R} E. H. Lieb and  D.W. Robinson \newblock
``The finite group velocity of quantum spin systems"\newblock
\textit{Commun. Math. Phys} \textbf{28} (1972), 251-257.


\bibitem{Liu} W. Liu, \newblock
``The spectra of the surface Maryland model for all phases'', \newblock 
\textit{Proc. Amer. Math. Soc.} \textbf{144} (2016), 5035--5047.

\bibitem{Loomis} L. H. Loomis, \newblock
``A note on the Hilbert transform'', \newblock
\textit{Bull. Amer. Math. Soc.} \textbf{52} (1946), 1082–1086.

\bibitem{Mt-S} R. Matos and J. Schenker.\newblock ``Localization and IDS Regularity in the Disordered Hubbard Model within Hartree–Fock Theory", \newblock \textit{ Commun. Math. Phys.} \textbf{382} (2021) 1725–1768.


\bibitem{M-S-J} R. Mavi and J. Schenker, \newblock 
``Resonant Tunneling In A System With Correlated Pure Point Spectrum'', \newblock 
\textit{J. Math. Phys.} \textbf{60} (2019), 052103.

\bibitem{M-S} R. Mavi and J. Schenker, \newblock 
``Localization in the Disordered Holstein model'', \newblock
\textit{Commun. Math. Phys.} \textbf{365} (2018) 719-764.
\bibitem{Poltora} A. Poltoratski, \newblock
``On the distributions of boundary values of Cauchy integrals'',  \newblock
\textit{Proc. Amer. Math. Soc.} \textbf{124} (1996), 2455–2463.

\bibitem{SPZ} A. Poltoratski, B. Simon and M. Zinchenko, \newblock 
``The Hilbert Transform of a Measure'' \newblock 
\textit{J. Anal. Math.}, \textbf{111} (2010), 247-265.

\bibitem{R-S} C. Radin, B. Simon \newblock
``Invariant Domains for the
Time-Dependent Schr\"odinger Equation" \newblock
\textit{ Journal of Differential Equations}, \textbf{29}  (1978), 289-296.


\bibitem{Simoncyclic} B. Simon, \newblock
``Cyclic vectors in the Anderson model'', \newblock 
\textit{Rev. Math. Phys.} \textbf{06} (1994), pp. 1183-1185.

\bibitem{Simon-Book} B. Simon, \newblock
\textit{Harmonic Analysis: A Comprehensive Course in Analysis, part 3}, \newblock
Amer. Math. Soc., Providence, RI, 2015.

\bibitem{Simon} B. Simon, \newblock 
``Operators With Singular Continuous Spectrum, VI. Graph Laplacians And Laplace-Beltrami Operators'', \newblock 
\textit{Proc. Am. Math. Soc.} \textbf{124} (1996), 1177-1182.

\bibitem{Stein} E. M. Stein and G. Wiess, \newblock
``An extension of a theorem of Marcinkiewicz and some of its applications'', \newblock
\textit{J. Math. Mech.} \textbf{8} (1959), 263–284.

\bibitem{Stolz} G. Stolz, \newblock
``An introduction to the Mathematics of Anderson Localization'', \newblock
pp. 71-108 in \textit{Entropy and the quantum II}, \newblock 
Contemp. Math. \textbf{552}, Amer. Math. Soc., Providence, RI, 2011.

\bibitem{Teschl} G. Teschl, \newblock 
\textit{Mathematical Methods in Quantum Mechanics with Applications to Schr\"odinger Operators}, \newblock 
Graduate Studies in Math. \textbf{99},  Amer. Math. Soc., Providence, RI, 2009.
 




 \end{thebibliography}

 \end{document}